\documentclass[12pt]{amsart}

\pdfoutput = 1

\usepackage{macros}
\usepackage[normalem]{ulem}

\newcommand{\nocontentsline}[3]{}
\newcommand{\tocless}[2]{\bgroup\let\addcontentsline=\nocontentsline#1{#2}\egroup}

\bibliography{generic.bib}

\begin{document}
\onehalfspacing
	
\author{Fabian Hahner,$^\flat$ Surya Raghavendran,$^{\sharp,\triangle}$\\[0.2ex] Ingmar Saberi,$^\natural$ Brian R.~Williams$^\circ$}

\email{fhahner@uw.edu}
\email{surya.raghavendran@yale.edu}
\email{i.saberi@physik.uni-muenchen.de}
\email{bwill22@bu.edu}

\address{{$^\flat$}University of Washington, Department of Physics \\ 3910 15th Ave NE, Seattle, WA 98195,  U.S.A}

\address{{$^\sharp$}%
    Yale University, Department of Mathematics \\ 
    P.O.~Box 208283,
New Haven, CT 06520, U.S.A.}

\address{{$^\natural$}Ludwig-Maximilians-Universit\"at M\"unchen \\ Theresienstra\ss{}e 37, 80333 M\"unchen, Deutschland}

\address{{$^\circ$}%
Boston University, Department of Mathematics and Statistics\\
665 Commonwealth Ave,
Boston, MA 02215, U.S.A.}

\address{{$^\triangle$}%
    University of Edinburgh, School of Mathematics \\
    Peter Guthrie Tait Road, King's Buildings, 
Edinburgh, EH9 3FD, U.K.}

\title{Local superconformal algebras}

\begin{abstract}
Given a supermanifold equipped with an odd distribution of maximal dimension and constant symbol, we construct the formal moduli problem of deformations of the distribution. 
This moduli problem is described by a local super dg Lie algebra that provides both a resolution of the structure-preserving vector fields on superspace and a derived enhancement of superconformal symmetry. 
Applying our construction in standard physical examples returns the conformal supergravity multiplet in every known example, in any dimension and with any amount of supersymmetry---whether or not a superconformal algebra exists.
We discuss new examples related to twisted supergravity, higher Virasoro algebras, and exceptional super Lie algebras. 
The compatibility of our techniques with twisting also leads to a computation of every twist of the stress tensor multiplet of a superconformal theory, including universal operator product expansions.
Our approach uses a derived model for the space of functions constant along the distribution, which is applicable even when the distribution is non-involutive;
we construct other natural multiplets, such as K\"ahler differentials, that appear naturally through this lens on superspace geometry.
\end{abstract}
	
	\maketitle
	\thispagestyle{empty}

	\setcounter{tocdepth}{1}
	
        \newpage
	\tableofcontents
	
	
	\setlength{\parskip}{7pt}

        \centerline{\emph{In memory of Yuri Manin}}
	
\section{Introduction}

\subsection{Deformations of superconformal structures}
It is well-known that the infinitesimal conformal transformations of flat space in dimensions $d \geq 3$ behave in drastically different fashion than those in dimension two. 
The Lie algebra of conformal Killing vector fields is finite-dimensional, and is isomorphic to $\lie{so}(d+1,1)$. 
  This is significantly less rich than the infinite-dimensional algebra of conformal transformations in dimension two, where, after complexification, any holomorphic or antiholomorphic vector field is a conformal Killing vector field.  
The exceptional behavior of the two-dimensional case is related to numerous well-known phenomena and techniques in the physics of two-dimensional field theories, which are believed not to generalize to larger-dimensional examples.

Recently, Kapranov has shown that the dichotomy between the $d=2$ and $d \geq 3$ cases can be explained using the language of derived geometry \cite{KapranovConf}.
Indeed, on any conformal manifold~$M$, there exists a \textit{dg} Lie algebra $\cL_\conf(M)$ which can be taken to be concentrated in degrees zero and one.
In degree zero, the cohomology of this Lie algebra is exactly the Lie algebra of conformal Killing vector fields on~$M$.
In degree one, the cohomology of this Lie algebra is (locally) infinite-dimensional for $d \geq 3$. 
Through the formalism of derived geometry, degree one elements define first-order (infinitesimal) deformations of the chosen conformal structure on~$M$.
Thus, while the space of automorphisms of a conformal manifold is finite-dimensional for $d \geq 3$, the space of deformations is infinite-dimensional.

\emph{Superconformal algebras} are super Lie algebras that contain both infinitesimal conformal transformations and supersymmetries.
They enlarge the usual supersymmetry algebras of physical interest, which combine the usual symmetries of flat spacetime with spinorial odd translations, by adding in special conformal transformations and their supersymmetric partners, together with scale transformations and R-symmetries. The latter are outer automorphisms of the supersymmetry algebra that fix all elements of even degree. Superconformal field theories, which exhibit a symmetry by some superconformal algebra, are central objects of study in the modern literature.
(The reader unfamiliar with the literature on supersymmetry and geometric approaches to it is directed to~\cite{ FayetFerrara, Sohnius, Superspace, CdAF}, just for example, and to references therein. We cannot hope to give complete references here.)

Using a characterization based on the properties sketched above, superconformal algebras were classified by Nahm~\cite{Nahm} and also by Shnider~\cite{Shnider}, and in dimension two in work of Kac and van de Leur~\cite{KacSusy}. Their behavior in different dimensions is seemingly even more erratic than that of conformal transformations. They are infinite-dimensional if $d=2$, but for $d\geq 3$ they are finite-dimensional simple super Lie algebras~\cite{Kac}. They exist in infinite families in dimensions three, four, and six---
but in dimension five, there is only one exceptional example. For $d>6$, no superconformal algebra exists at all.

In this paper, we study superconformal algebras from a derived perspective. 
We introduce \emph{local superconformal algebras}, which are super dg Lie algebras $\Conf(\fn)$ (Definition~\ref{dfn:Conf(n)}) defined on any space equipped with a \emph{superconformal structure}. 
We recall the
the notion of such a structure, and then come up with an appropriate formalism for studying its geometry and its symmetries. As in Kapranov's approach---indeed, as is always the case in derived deformation theory---symmetries, deformations, and potential (higher) obstructions fit together into the single controlling object $\Conf(\fn)$. 

Interestingly (and importantly), a superconformal structure has \emph{a priori} nothing to do with an equivalence class of metrics. It is given by an odd subbundle of the tangent bundle of a supermanifold satisfying certain conditions (Definition~\ref{def:superspace}). 
The notion is not new; to the best of our knowledge, the definition in this form (in four spacetime dimensions) was first given by Manin exactly forty years ago in~\cite[chapter 5, \S7]{Manin}, abstracting work of Ogievetsky and Sokatchev~\cite{OS} that was reformulated and extended in Schwarz' important paper~\cite{SchwarzSugra}.
Such a geometric datum is present in any superspace approach to supersymmetric field theory. It is normally given by specifying formulas for the ``supercovariant derivatives'' $D_a$, or equivalently for the odd vector fields that implement supersymmetry transformations. For this reason, and for brevity, we refer to a supermanifold equipped with a superconformal structure as a \emph{superspace}.  
This paper constructs the dg Lie algebra controlling the formal moduli problem of deformations of superspace. 

We emphasize that \emph{there is no restriction whatsoever} on the dimension, or on the type of supersymmetry transformations, for which a superconformal structure can be defined. Our approach generalizes to superconformal structures based on any two-step nilpotent super Lie algebra, and in principle to any sort of geometry defined by a choice of tangential distribution of constant symbol. As such, our perspective restores a pleasing uniformity to the story.

There are three immediate questions the reader may have in mind:
\begin{itemize}[label={---}]
\item How does the local superconformal algebra relate to Nahm's classification?
\item How do superconformal structures relate to normal conformal geometry?
\item Are deformations of superconformal structures related to conformal supergravity?
\end{itemize}
We summarize our answer to each in turn.

\numpar[intro:nahm][The relation to Nahm's classification]

The symmetries of a superconformal structure are captured by the zeroth cohomology of $\Conf(\fn)$. Evaluating on the flat superspace, this recovers the known superconformal algebras (realized in terms of supervector fields) whenever they exist and the super Poincar\'e algebra, including R-symmetry, together with scale transformations in all other cases.

There is in fact a general technique for constructing the (super) Lie algebra of symmetries for geometries of this type, known as \emph{Tanaka prolongation}. The maximal prolongation of the supertranslation algebra $\fn$ recovers the symmetries of the flat superspace based on~$N$. In their beautiful paper~\cite{AltomaniSanti},
Altomani and Santi apply this technique to (physical) supertranslation algebras, recovering the classification of superconformal algebras in a new fashion. Their work proves that the maximal prolongation of such a supertranslation algebra is generically just its extension by automorphisms $\fg_0$. Superconformal algebras are ``exceptional prolongations,'' and exist in each case due to accidental isomorphisms.

By construction, $H^0(\Conf(\fn))$ consists of the symmetries of the superspace on which we evaluate it. On flat superspace, we recover the maximal prolongation of~$\fn$. Thus the zeroth cohomology of our local superconformal algebra on flat space \emph{is} the standard superconformal algebra wherever it exists.

We can apply base change to $H^0(\Conf(\fn))$ to obtain smooth super vector fields of the conventional sort. Doing this, our results recover geometric representations of superconformal algebras using conformal Killing supervector fields. (For a recent discussion, we refer to~\cite{ConformalSuperKilling} and also to~\cite{HoweLindstrom}; the original treatments go back to Sohnius in four dimensions~\cite{SohniusSUCO} and Park in dimensions three~\cite{Park3d} and six~\cite{Park6d}.)

\numpar[intro][Superconformal structures and conformal structures]
One might well ask what led Manin to call such a datum a ``superconformal structure.''
The answer is that a choice of superconformal structure simultaneously effects
a reduction of the structure group of the underlying manifold to the group $G_0$ arising from automorphisms of the supersymmetry algebra.
The reduction is along the map $\rho_2: G_0 \to GL(d)$ given by the action in even degree.
For standard supersymmetry algebras, the image of~$\rho_2$ is isomorphic to $SO(d) \times \R_+$, so that this reduction of the structure group determines (in particular) a conformal structure. 
We make this intuition precise in Theorem~\ref{thm:s/conf} below, which
constructs a comparison map from the even subalgebra $\Conf(\fn)_+$ to a dg Lie algebra $\cL_{\fg_0}$
that describes the moduli problem of manifolds equipped with a reduction of the infinitesimal structure group to~$\fg_0$.
In the standard case, we can forget the kernel of~$\rho_2$ (which is related to a choice of principal R-symmetry bundle), obtaining a comparison map
\deq[eq:phi]{
	\phi \colon \Conf(\fn)_+ \to \cL_\conf
}
between even-parity superconformal structures and conformal structures. The name is thus sensible in the context of the standard examples. We comment further on this below.

\numpar[p:confsugra][Conformal supergravity] 
One of the primary motivating problems in supergeometry and supergravity has been to understand supergravity multiplets as related to natural moduli problems of geometric structures on superspace.
To sketch the analogy: The corresponding insight in the theory of general relativity is that the physical degrees of freedom describing the gravitational field correspond precisely to deformations of the pseudo-Riemannian metric on spacetime. (The Einstein--Hilbert action then specifies \emph{dynamics} for these degrees of freedom, but this can be thought of as a second step, after the kinematic nature of the degrees of freedom has been understood.)

Since we construct the moduli problem of deformations of superconformal structures, it is natural to guess that the output of our computation should be related to the conformal supergravity multiplets that are studied in the literature. Indeed, this is the case. In \S\ref{sec: examples}, we exhibit the component-field models of $\Conf(\fn)$ for all important physical examples. In each case, we find precise agreement with the conformal supergravity multiplets described in the literature. We thus take a substantial step towards giving a uniform, example-independent geometric formulation of supergravity; at the level of kinematical conformal supergravity, this program is complete.

We emphasize that our method for constructing these multiplets has three distinct advantages. Firstly, it is \emph{algorithmic}. The only input is the supertranslation algebra, and many computations are simple; the component fields can typically be identified in a fraction of a second using open-source commutative algebra software such as \emph{Macaulay2}~\cite{M2}. Secondly, it is \emph{geometric}. Our construction is designed to describe the moduli problem of deformations of a superconformal structure, so this interpretation is automatic. Thirdly, it is \emph{complete:} we recover not only the multiplet, but a full $L_\infty$ structure specifying the complete structure of all linear and nonlinear gauge transformations, together with whatever ``differential constraints'' are required by the structure of the supersymmetry representation. 

We do \emph{not} go on to discuss dynamical theories of conformal supergravity. The output of our construction should (in general) be thought of as an \emph{off-shell} datum. Among the examples we study, only one notable instance---eleven-dimensional supergravity---leads to the appearance of an on-shell multiplet.\footnote{It is an amusing rule of thumb that the multiplet associated to the structure sheaf of spacetime becomes on-shell when sixteen supercharges are present, whereas the multiplet associated to the tangent sheaf of superspace does so only with thirty-two supercharges.}
Our results could, of course, be profitably be used for these purposes, and we look forward to pursuing this further in future work.

\subsection{Superspace geometry and  ``pure spinor'' constructions}
Our construction of local (derived) superconformal algebras 
uses a perspective on superspace geometry which arises out of recent work that revisits and geometrizes the pure spinor formalism.\footnote{``Pure spinor'' is somewhat misleading terminology, but is standard. As we will see, the same is true of ``conformal supergravity,'' and to some extent of ``superconformal structure.''} Pure spinor techniques go back nearly to the time of Manin's definition; the idea was first explored in the literature by Nilsson~\cite{Nilsson} and Howe~\cite{HowePS1,HowePS2}, and is perhaps most well-known in connection with Berkovits' approach to the superstring~\cite{Berkovits}. Here, we work only in the context of field theories, and thus are closest to the work of Cederwall and collaborators~\cite[and references therein]{Cederwall}.

The connection of the pure spinor formalism to the geometry of superspace has recently become clear.
Building on~\cite{spinortwist,MSJI} and especially~\cite{CY2}, we take the perspective that \emph{the most appropriate structure sheaf for superspace is not the smooth functions}. 
Taking smooth functions as the structure sheaf, as is often done implicitly in the literature, ignores the geometric datum provided by the distribution. 
It is analogous to studying the sheaf of all smooth functions on a complex manifold. Clearly, it is more appropriate to study the sheaf of \emph{holomorphic} functions, which are those functions that are annihilated by sections of the distribution (in this analogy, the antiholomorphic tangent bundle $\overline{\T} \subset \T_\C$). 

Since typical superconformal structures (in contrast to complex structures) are bracket-generating, the only functions that are constant along the distribution are constants, and one would naively say that nothing can come of the idea.
A common approach in dimension four is to identify an involutive subdistribution, and to impose invariance only there. The relevant subdistributions are defined by chiral spin bundles, and the procedure outputs ``chiral superfields.'' (For our perspective on this construction, see~\S\ref{chiral}.) 

In general, though, no appropriate involutive subdistributions exist. The literature is replete with various other approaches to dealing with these ``torsion constraints'' in particular examples,  and their central importance has been widely recognized.
Our approach is to simply take the \emph{derived} invariants of sections of the distribution. This is analogous, in complex geometry, to constructing the Dolbeault complex, rather than imposing holomorphy strictly.
        
In~\cite{CY2}, building on the beautiful construction of Dolbeault cohomology for almost-complex manifolds given by Cirici and Wilson in~\cite{ACDolbeault}, a derived model for functions on superspace was constructed in this manner, using the filtration of the de Rham forms defined by the superconformal structure. 
The analogue of the Dolbeault complex
is nothing other than the scalar pure spinor superfield, called the ``canonical supermultiplet'' in~\cite{MSJI} and the ``tautological filtered cdgsa'' in~\cite{spinortwist}. In general, this complex has cohomology in nonzero degrees even on flat space. Its cohomology recovers the physical fields of various important multiplets. (See Table~\ref{tab:O} in~\S\ref{mu} below for some examples.) We regard it as the structure sheaf of superspace.

        On general grounds, the moduli problem describing formal structure-preserving deformations of a ringed space is controlled by the tangent sheaf, which admits a local description as the derivations of the structure sheaf.\footnote{In general, one should in fact work with the tangent complex. We reserve ``derived superconformal algebra'' for this object; see~\S\ref{p:derived}.}
        This moduli problem is the local superconformal algebra. We will
        study the resulting supermultiplet in detail for superspaces of physical interest, matching it in each instance to the full conformal supergravity multiplet.
But our construction goes beyond these examples. We study further instances, related to twisted theories and holomorphic field theories.
        Fleshing out our geometric perspective, we also define and study other natural supermultiplets associated to the geometry of superspace, such as the sheaves of K\"ahler differentials and differential operators.

Our approach connects cleanly to the mathematical theory of \emph{non-holonomic $G$-structures} \cite[for example]{Tanaka1,Tanaka2,CapSlovak,Zelenko,AD}. 
The relationship of this body of work to supersymmetry has been appreciated and developed by numerous groups, notably in important work of Santi and collaborators~\cite{SantiSpiro1,SantiSpiro2,SantiFOF1,SantiFOF2,KST}. For more related work on ideas from Cartan geometry in the context of gravity, see~\cite{CartanSugra} for an excellent recent survey with comprehensive references to the literature, including the work of the Torino school~\cite{CdAF}, or \cite{Wise} for a very readable introduction in the bosonic case.

In comparison with these works, the novel ingredients in our setting come from derived geometry. 
In the broadest terms, we view our program as aimed at constructing and understanding the derived geometry of non-holonomic $G$-structures.
In principle, our methods extend beyond superspaces to general Tanaka structures on (super)-manifolds; for an example of some closely related constructions, see~\cite{BGGcpx}. 
We plan to elaborate the connection between pure spinor techniques and complexes of Bernstein-Gelfand-Gelfand type in future work.

        \subsection{Physical interpretation and applications}
        We give a few more informal remarks to contextualize our results, aiming primarily at physics-minded readers. The reader should feel free to skip this section entirely, or to refer back to it for commentary after perusing the main constructions, according to taste.

\numpar[p:confstr][What about standard supergravity?]
It is natural to ask why only \emph{conformal} supergravity multiplets appear in this context, rather than (for example) the full supergravity multiplet. This occurs because every supertranslation algebra is consistently $\Z$-graded, so that scale transformations always appear in the structure group of a $G$-structure that arises from a superconformal structure. This is perhaps the best possible justification for the name; a better name would have perhaps have emphasized local scale invariance. No matter the name, the important object is the moduli space of deformations of a geometric structure defined by a tangential distribution, and the relation to conformal structures is an output rather than an input.

It is natural to ask how to extend these results to give formulations of more standard supergravity multiplets or theories. The essential clue is provided already in~\cite{Manin} and in~\cite{Deligne}, where it is pointed out that one formulation of a ``super Riemannian structure'' consists of a superconformal structure together with an additional datum given (roughly) by some chosen section of the Berezinian.
The additional datum plays the role of the choice of a specific metric within a conformal class.

One then expects that the normal supergravity multiplet stands in the same relation to the \emph{divergence-free} vector fields on superspace as the conformal supergravity multiplet does to all vector fields on superspace. 
From our perspective, this is one intuitive reason for the ubiquity of techniques involving conformal supergravity and ``compensator fields'' in the literature.
It should be relatively straightforward to give a uniform, example-independent construction of the conformal compensator, extending $\Conf(\fn)$ by the sheaf of sections of the Berezinian. Doing so would complete a uniform, example-independent geometric approach to supergravity at the kinematic level.
We look forward to revisiting these ideas and intuitions in more detail in future work.

We remark that, at the level of BV theories, we \emph{can} straightforwardly understand a mechanism analogous to the conformal compensator in our formalism in the cleanest example, eleven-dimensional supergravity; see the discussion in~\S\ref{p: 11d}.

\numpar[p:moduliprob][Superconformal field theories]
We make a few quick remarks on the relation of our family of super dg Lie algebras to superconformal field theories. 
Our viewpoint is informed by the perspective of Costello and Gwilliam, as developed in~\cite{CG1,CG2}, and especially by their emphasis on the connections between field theory, factorization algebras, and {derived deformation theory}. 

Implicitly, it is common practice to think of the action of a symmetry on a theory in terms of coupling to a corresponding family of backgrounds. This 
maneuver typically goes via some version of Noether's theorem.
For example, the current $J$ is constructed from the action of a global symmetry on the theory---but also specifies the coupling to a background connection, via the coupling term $\int A \wedge J$.

Within derived geometry, the symmetries and the backgrounds are encoded in terms of a single formal moduli problem, described by a local Lie algebra $\cL$. 
Given such a local Lie algebra, Costello and Gwilliam define an associated \emph{current algebra}~\cite{CG2}. This is the
factorization algebra
\beqn
\Cur(\cL) \define C_\bu(\cL_c) .
\eeqn
More explicitly, this factorization algebra assigns to an open set $U \subset M$ the cochain complex
\beqn
C_\bu (\cL_c(U))
\eeqn
where $\cL_c(U)$ denotes the Lie algebra of sections of $\cL$ whose support is compact in $U$, and where $C_\bu(-)$ is the Chevalley--Eilenberg functor which computes Lie algebra homology.

For newcomers, it is intuitively correct to imagine a factorization algebra as a structure that coherently and rigorously encodes 
the structure of the operator product expansion. The essential examples come from the \emph{observables} of a perturbative field theory $\cT$ and from the current algebra construction mentioned above. 

Costello and Gwilliam prove a far-reaching generalization of Noether's theorem at the level of factorization algebras. Given a theory $\cT$ with a symmetry by~$\cL$, their result constructs a map 
\deq{
\Cur(\widehat{\cL}) \to \Obs(\cT)
}
of factorization algebras. The notation $\widehat{\cL}$ reflects the presence of a canonically determined local central extension, reflecting the anomaly of the symmetry and measured by a class in $H^1_\loc(\cL)$. (The cohomological approach to anomalies is familiar in the physics literature, and dates back at least to~\cite{BPT-ABJ}.)

At the level of coupling to backgrounds, a supersymmetric theory is a theory that can be defined on the class of superspaces of a particular type, possibly equipped with some extra data (such as the section of the Berezinian mentioned above). A supersymmetric theory is \emph{superconformal} when it depends only on a superconformal structure, and is thus naturally defined on the class of all superspaces \emph{without} additional structure. The factorization Noether theorem then equips the observables of a superconformal field theory with a map from the factorization algebra of currents of $\Conf(\fn)$. The image of this map is the stress tensor multiplet of the theory---together with all universal OPEs, conservation laws, and improvement transformations, which are encoded in the factorization algebra $\Cur(\Conf(\fn))$.

As such, one can view our results dually as a uniform and example-independent construction of superconformal stress tensor multiplets at the level of factorization algebras. We have this construction in mind in~\S\ref{sec: twisted}, where we are interested in computing the twists of stress tensor multiplets. Our results there show the connection between four-dimensional $\N=2$ superconformal theories and higher Virasoro algebras cleanly and directly, and provide a way to understand the full derived symmetry algebra appearing in any twist of any superconformal theory. 

There is, of course, work to be done understanding anomalies using our technology. Such work would recover, for example, the central-charge identities of~\cite{BeemEtAl}---as lifted to the higher Virasoro algebra in~\cite{SCA}---by constructing an appropriate comparison map between $\Conf(\fn)$ and its twist $\Conf(\fn_Q)$. Some work of this sort in the context of six-dimensional $\N=(2,0)$ supersymmetry will appear shortly~\cite{SW-E36}.

\subsection{Structural overview}
We briefly sketch the organization of the paper, pointing out the essential results of each section.

In~\S\ref{sec: parabolic}, we review and develop 
our perspective on superspace geometry, constructing the structure sheaf $A^\bu$ of a superspace and the sheaves of holomorphic forms on it (\S\ref{geom}), 
 the corresponding local superconformal algebra $\Conf(\fn)$ (\S\ref{der-A}), and the sheaf of K\"ahler differentials (\S\ref{p:Kaehler})). We describe the sheaves on pure spinor space that produce each of these multiplets, in particular proving that $\Conf(\fn)$ is the multiplet associated by the pure spinor formalism to the sheaf of surviving translations on superspace (Theorem~\ref{thm: coker}). We furthermore construct a map from the sheaf of one-forms to the sheaf of K\"ahler differentials, and comment on some subtleties in the relationship~(\S\ref{p: oneforms}). 

In~\S\ref{s:universal}, we give some reminders about comparing our models to component-field models in physics, and characterize the cohomology and the component fields of $\Conf(\fn)$ in low degrees. The main results are Proposition~\ref{prop:prolong}, an observation relating the zeroth cohomology to the maximal transitive prolongation of~$\fn$ (and thus to Nahm's list), and Theorem~\ref{thm: univ}, which---upon specializing to physical examples---implies the universal presence of smooth vector fields, local supersymmetry transformations, and $R$-symmetry transformations in degree zero, together with the vielbein (and its superpartners) in degree one.

In~\S\ref{sec: Lie}, we go on to characterize the local Lie algebra structure of this portion of $\Conf(\fn)$ explicitly. The central result here (Theorem~\ref{thm:s/conf}) constructs the comparison map from the moduli problem of superconformal structures to the moduli problem of manifolds with $G$-structure, where $G$ is the automorphism group of the supertranslation algebra. This specializes in the standard examples to give a map to the moduli problem of conformal structures (Corollary~\ref{cor:s/conf}).

We then go on to explicitly work out the standard physical examples in~\S\ref{sec: examples}.
We
fully compute the chain complex of vector bundles of the component-field model $\mu \Conf(\fn)$, working case by case and showing concretely that each agrees with the known conformal supergravity multiplets at the level of the component-field multiplet. The reader who is not specifically interested in these particular examples may skip this section, which contains no structural result; on the other hand, the reader familiar with the details of conformal supergravity multiplets may find it helpful to begin here.

Finally, in~\S\ref{sec: twisted}, we explore a few examples related to twisted theories and twisted supergravity. 
\S\ref{ssec:tw4d} treats holomorphic twists of four-dimensional theories. To clarify the relationship between our models and extensions of higher Virasoro algebras, we comment on some issues related to chiral superspace in~\S\ref{chiral}, and identify $\Conf(\fn)$ for holomorphically twisted chiral superspace with the Lie algebra of holomorphic super vector fields. 
\S\ref{ssec:tw6d} then treats holomorphic twists of six-dimensional theories; the essential example is $\N=(2,0)$ supersymmetry, for which the local superconformal algebra on flat space recovers the exceptional simple super Lie algebra $E(3|6)$.

\subsection{Acknowledgements}
We gratefully acknowledge Martin Cederwall, Richard Eager, Chris Elliott, Simon Jonsson, Simone Noja, Jakob Palmkvist, and Johannes Walcher for many conversations about, ideas in, and contributions to the program that is continued here. IAS gives particular thanks to John Huerta for conversations and inspiration, including many essential insights about superconformal structures and Tanaka prolongation, in the course of ongoing joint work on related topics, and to Owen Gwilliam for clarifying conversations both about the work presented here and about the program more generally. We further thank E. Bergshoeff, I.~Brunner, K.~Costello, and N.~Paquette for conversations,
and especially Jos\'e Figueroa-o'Farrill for correspondence.
SR thanks the Deutsche Forschungsgemeinschaft (DFG) under Germany's Excellence Strategy EXC-2094 390783311 (Excellence Cluster ORIGINS) for hospitality.
IAS thanks the LMU Center for Advanced Studies for its support during the winter semester 2023/24, where he was in residence as a Junior Researcher as this work was being completed, and further thanks Perimeter Institute for Theoretical Physics, the Deutsches Elektronen-Synchrotron, the University of British Columbia, the Wilhelm and Else Heraeus-Stiftung, the University of Idaho, the Center for Rural Livelihoods, the Max-Planck-Institut f\"ur Mathematik, and Freddy's Caf\'e in Kolding for further hospitality and support during its preparation.
BRW also thanks the LMU-CAS for its hospitality in January 2024.
This work is funded by the Deutsche Forschungsgemeinschaft (DFG, German Research Foundation) under 
Projektnummer 517493862 (Homologische Algebra der Supersymmetrie: Lokalit\"at, Unitarit\"at, Dualit\"at). 
IAS is further supported by the Free State of Bavaria. FH is supported by funds from the DOE Early Career Research Program under award DE-SC0022924.

\section{Superconformal structures and their symmetries} 
\label{sec: parabolic}

In this section, we review the geometric perspective on pure spinor superspace, as articulated in~\cite{CY2}, as well as those technical aspects of the pure spinor superfield formalism we will need. For details on our approach to the formalism, the reader is referred to~\cite{perspectives,derivedps}, and for further literature to~\cite{Cederwall} and references therein.

\subsection{Superspaces and tangential distributions}
\label{sec: superpara}
We begin by reviewing the definition of a superspace, which is a supermanifold equipped with the extra datum of a tangential distribution $D$, of maximal odd dimension, which fails to be involutive in a prescribed, locally constant manner. This datum is called a \emph{superconformal structure}~\cite{Manin,Deligne}, though it is important to emphasize that it does not
reduce to the notion of a conformal structure if the supermanifold is purely even. One main aim of this paper will be to unpack the relationship to conformal structures further by studying deformations of superconformal structures, and constructing, under certain additional assumptions, a comparison map from \emph{even} deformations of superconformal structures to deformations of conformal structures (Theorem~\ref{thm:s/conf}).

\numpar[tanaka][Regular tangential distributions]
Our definitions here follow~\cite[\S2.2]{AD2}, to which we refer the reader for details.
Let $M$ be a smooth manifold and $D \subset TM$ a distribution.
To any point $p \in M$, we can associate a weight-graded Lie algebra $\fn_p$, called the \emph{symbol algebra} at $p$ as follows. 
$D$ defines a filtration 
\deq{
0 \subset D_p =  F^1 T_p M  \subset F^2 T_p M \subset \cdots
}
on each vector space $T_p M$ by taking $F^k T_p M$ to be the span of the evaluation at $p$ of all vector fields obtained by $\leq k$-fold brackets of sections of~$D$ in a neighborhood of~$p$. The filtration stabilizes after a finite number of steps. $D$ is called \emph{bracket-generating} if the filtration is exhaustive; if it is not, we extend it to an exhaustive filtration by placing $T_p M$ at the step after the filtration stabilizes.

If we extend vectors $v,v' \in D_p$ to sections $X_v,X_{v'}$ of $D$ in some neighborhood $U$ of $p$, the restriction of the vector field $[X_v,X_{v'}]$ to the point $p$ will depend on the extensions. Thus the filtered Lie algebra structure on vector fields over~$U$ does not induce one on $T_p M$. However, passing to the associated graded defines a grade-preserving Lie bracket on $\fn_p = \Gr T_p M$. This is the symbol algebra.

The notion of the symbol generalizes straightforwardly to the case of supermanifolds, where an additional $\Z/2\Z$ grading by parity is present. In the examples we will later be interested in, the parity is determined by the weight grading modulo two, but this is not essential.

Fix a finite-dimensional (super) Lie algebra $\fn$ that is positively weight-graded, and let $G_0$ denote the group of automorphisms of~$\fn$ as a weight-graded Lie algebra. A distribution is called \emph{regular of type $\fn$} if the symbol algebra $\fn_p$ is isomorphic to~$\fn$ for all $p\in M$. A choice of isomorphism
\deq{
\psi : \fn \to \Gr T_p M
}
is called an \emph{adapted frame} at~$p$. The collection of adapted frames forms a principal $G_0$-bundle.

\numpar[p:examples][Examples]
Here is a list of running examples of such distributions that the reader can keep in mind:
\begin{enumerate}[label=\arabic*~---]
\item An almost-complex structure on a smooth manifold $M$ is the same datum as a splitting 
\deq{
T_\C M = T \oplus \bar{T}
}
of the complexified tangent bundle as a sum of conjugate subbundles. We take the distribution $D$ to be given by~$\bar{T}$. Given vectors $v,v' \in D_p$ and extensions to sections $X_v,X_{v'}$ of $D$, the symbol algebra at~$p$ is defined by the Nijenhuis tensor
\deq{
\gamma: \wedge^2 \bar{T} \to T_\C M/\bar{T} \cong T, \quad
(v, v') \mapsto  [X_v, X_{v'}](p) + D_p.
}
The distribution is regular if the Nijenhuis tensor has constant rank.
\item A complex structure is the special case of the above with $\gamma = 0$. It is a regular distribution of type $\fn$, with $\fn$ taken to be the abelian weight-graded Lie algebra with $\overline{\C}^n$ in weight one and $\C^n$ in weight two.
\item A transversely holomorphic foliation on a smooth manifold of dimension~$d$ is determined as follows: Choose a subspace $\overline{A} \subset \C^d$ such that $\overline{A}$ and its conjugate ${A}$ together span $\C^d$. Then define $\fn$ to be the abelian weight-graded Lie algebra with $A$ in degree one and $\C^d / A$ in degree two. A THF structure is a regular distribution of type~$\fn$ in the complexified tangent bundle. It would be interesting to consider generalizations to regular noninvolutive distributions.
\item A contact structure is a bracket-generating distribution of codimension one. It is a regular distribution of type $\lie{h}$, where $\lie{h}$ denotes the Heisenberg Lie algebra: the central extension of $\R^{2n}$ specified by a symplectic form.
\item Any supermanifold on which a supersymmetric field theory can be defined. See~\S\ref{p:superspace}.
\end{enumerate}

\numpar[p:kleinmodel][Model geometries]
There is a standard model geometry (in the sense of Cartan geometry) for manifolds equipped with regular distributions of type~$\fn$. 
It is sometimes called the ``flat distribution of type $\fn$.'' One constructs it by considering the simply connected (super)group $N = \exp(\fn)$ which exponentiates the symbol algebra. (More properly, flat superspace is a torsor over~$N$, but we will abusively conflate the two.)  
Since $\fn$ is nilpotent, $N$ is topologically just a (super) vector space.

The right-invariant vector fields extending $\fn_1 = D_0 \subset T_0 N$ span the distribution~$D$. 
It is immediate by translation invariance that the symbol of $D$ is $\fn$ at each point of~$N$. It is furthermore clear that the left action of~$N$ on itself is compatible with the distribution, and that it extends to the action of the larger Lie group $G = G_0 \rtimes N$.  So we are considering a version of Klein geometry for the pair $(G, G_0)$. (For this notion and an excellent introduction to related ideas, see~\cite[definition 3.16]{Sharpe}.)
Note, though, that the distribution-preserving symmetries of~$N$ can be \emph{larger} than~$G$; to describe them algebraically is the purpose of Tanaka's prolongation procedure.

\numpar[sec:ST][Supertranslations]
We now specialize to the class of examples we will focus on in the following. 
To begin, we fix what we call a \textit{supertranslation algebra}~$\lie{n}$.
By definition, this is a finite-dimensional, consistently weight-graded
super Lie algebra supported in weights one and two. (See~\cite{Kac} for basic definitions.)
Thus $\lie{n} = \lie{n}_1 \oplus \lie{n}_2$ as a vector space, $\lie{n}_1$ is odd and $\lie{n}_2$ is even, and the only nontrivial bracket is given by a single linear map
\deq{
    \gamma: \Sym^2(\fn_1) \to \fn_2.
}
For now, we make no further requirements on~$\fn$.
At this level of generality we are able to consider both all standard examples (\S\ref{sec:stdST})
and more exotic superspaces, such as those obtained from the process of twisting, in a uniform way. 

In the sequel, we will write $d = \dim(\fn_2)$ and $k = \dim(\fn_1)$, so $\dim(\fn) = d|k$ as a super vector space. 
We emphasize that the $\Z$-grading here is \emph{distinct} from the cohomological grading; see \S\ref{sec:grading} below for grading conventions.

\numpar[sec:aut][Automorphisms]
We can extend $\fn$ by the super Lie algebra $\fg_0$ of its degree-zero derivations. (Note that there are no inner degree-zero derivations.) We will denote this extension by 
\beqn
\fg = \fg_0 \ltimes \fn .
\eeqn
$\fg$ plays the role of the affine transformations: for example, if $\fn_1 = 0$, then $\fg$ is precisely the Lie algebra of infinitesimal affine transformations of $d$-dimensional space.  

$\fg_0$ comes equipped with a pair of maps
\deq{
\rho_i: \fg_0 \to \lie{gl}(\fn_i), \quad i = 1, 2,
}
encoding the action of infinitesimal automorphisms in each degree $i$.
The ideal $\fr= \ker(\rho_2)$ is called the \emph{R-symmetry algebra}. Since $\fg$ is graded, there is always a $\lie{gl}(1)$ subalgebra $\fz$ of $\fg_0$, not contained in $\fr$, that makes the grading inner.

In the following, we will often use the notation $V = \fn_2$ and $\Sigma = \fn_1$ when we want to refer to these vector spaces just as $\fg_0$-representations, rather than as direct summands of~$\fn$.

\numpar[p:superspace][Superspaces and superconformal structures]
We can now remind the reader of an important definition:
\begin{dfn}[Well-known under diverse names: ``superconformal structure'' already in~{\cite[chapter 5, \S7]{Manin}}; ``SUSY manifold'' in~{\cite[lecture 2]{BernsteinSUSY}}, for example]
Let $\fn$ be a supertranslation algebra.
 A \emph{superspace based on~$\fn$} is a smooth supermanifold $M$ of dimension $d|k$, equipped with a \emph{superconformal structure} of type~$\fn$.
 A \emph{superconformal structure} of type~$\fn$ is a regular tangential distribution 
    \deq{
        D \subset TM
    }
    of type $\fn$---in other words, an odd distribution of maximal dimension whose symbol algebra is $\fn$ everywhere in~$M$.
    \label{def:superspace}
\end{dfn}
Again, there is a model geometry $N$ for superspace, which we think of as the Klein geometry $(G = G_0 \ltimes N, G_0)$. We will call this geometry \emph{flat superspace}. (Names abound; in~\cite{Supersolutions}, this is called ``super Minkowski space.'') In the sequel, we will always work on flat superspace, though we emphasize that our constructions and our perspective are in principle global.

Physicists should imagine that the distribution $D$ is locally spanned by the odd vector fields $D_a$ normally called the ``supercovariant derivatives.'' These fail to commute, and the failure is encoded by the structure constants $\gamma$ of the supertranslation algebra $\fn$ appropriate to the dimension and the amount of supersymmetry. This is the structure that is being specified when formulas for the supercovariant derivatives are written down.

To connect very explicitly to the standard literature, we recall that a standard set of coordinates on flat superspace consists of bosonic coordinates $x^\mu$ and fermionic coordinates $\theta^a$, determined by a corresponding choice of basis for~$\fn_2$ and $\fn_1$, respectively. 
The left-invariant vector fields extending $\fn_1 \subset T_0 N$ take the well-known form
        \deq{
            Q_a = \pdv{ }{\theta^a} + \theta^b \gamma_{ab}^\mu \pdv{ }{x^\mu}.
        }
The right-invariant vector fields take the form
        \deq{
            D_a = \pdv{ }{\theta^a} - \theta^b \gamma_{ab}^\mu \pdv{ }{x^\mu},
        }
recovering the standard formulae for the supersymmetry-covariant derivatives. Left- and right-invariant vector fields commute.
The Lie subalgebra of $\Vect(N)$ consisting of left-invariant vector fields is of course isomorphic to $\fn$; the even left-invariant derivatives are just $\partial/\partial x^\mu$.
We emphasize again that no additional structure---in particular, nothing about spinors---is of any relevance to our constructions, though we do discuss the standard examples in great detail.

\numpar[sec:stdST][Standard examples]
In the cases usually considered in the study of supersymmetric physics, the abelian Lie algebra $\fn_2 = V$ of translations is equipped with a symmetric inner product, and $\fn_1$ is a spinorial $\lie{so}(V)$ representation. The map $\gamma$ is constructed from the Clifford multiplication map by using the inner product and the spin-group-invariant pairing on the spin representation; the details depend on the dimension modulo eight, and on the signature. (Readers unfamiliar with supersymmetry algebras are referred to the review in~\cite{NV}, to~\cite{Supersolutions}, or to any of innumerable accounts in the literature.)

For algebras of this kind, $\fg_0$ is of the form $\so(V) \oplus \fr \oplus \lie{gl}(1)$; the factors represent infinitesimal Lorentz and R-symmetry, together with the grading, which physically corresponds to scaling dimension. It is the form of $\fg_0$ in these examples which creates a relationship between such classes of superconformal structures and conformal structures, via a reduction of the structure group to~$G_0$.
We will often denote the spin representation of $\Spin(V)$ by~$S$, and the chiral spin representations for even $d$ by $S_\pm$.

\subsection{Tangential distributions in derived geometry}
We develop an approach to the following general question: Given a space equipped with a tangential distribution, what is the sheaf of functions that are required to be constant along the distribution? On the face of it, the question is not interesting in our example: Since sections of the distribution generate all vector fields via the Lie bracket, the only functions that are annihilated by all sections of the distribution are constants. But this is too hasty. By using an appropriate derived model for the invariance condition, we find interesting commutative differential graded algebras with cohomology in nonzero degrees. Our approach is an abstraction of Cirici and Wilson's profound recent work generalizing the Dolbeault complex to almost-complex manifolds~\cite{ACDolbeault}.

We interpret this as saying that, for superspaces or other geometries defined by noninvolutive tangential distributions, the model geometry is most properly viewed as a derived stack. Even though it is fully described by its global algebra of functions, this derived stack cannot be thought of as an affine space in the normal sense. The correct derived model for functions annihilated by sections of the distribution, obtained via the corresponding filtration on the de Rham complex, always has cohomology in nonzero degree. We will see some examples of this that help to build intuition later on.

\numpar[p:analogy][Naming conventions; an important analogy]
Throughout, we will often use language suggested by the analogy between superspaces and almost-complex manifolds. ``Holomorphic,'' for example, will mean ``invariant (in a derived sense) along the distribution.'' We will try to use scare quotes to emphasize this metaphorical usage, but may not be entirely consistent throughout.

Our principal reason for foregrounding almost-complex structures, rather than contact structures or other geometric structures defined by tangential distributions, is experience with twists of supersymmetric field theories. These twists often output holomorphic (or holomorphic-topological) field theories, which are defined on spaces with complex structures or transverse holomorphic foliations. Because our construction is uniform, in a precise sense, across all possible twists of a fixed supersymmetry algebra (\S\ref{twists}), we often find descriptions of supersymmetric field theories that agree with the descriptions of their twists as \emph{holomorphic} field theories. Eleven-dimensional supergravity is a notable example~\cite{CY2}, as is the description of ten-dimensional Yang--Mills theory as a ``holomorphic'' Chern--Simons theory~\cite{BerkovitsSuperparticle}. The analogy is instructive, and so we choose to emphasize it throughout, at the cost of potentially introducing yet another disorienting linguistic choice into a subject already replete with them.

\numpar[sec:grading][Grading conventions]
Throughout, our conventions and terminology largely adhere to those in~\cite{MSJI}, but we recall the relevant parts here.
As mentioned above, we work with two $\Z$-gradings, with the Koszul sign rule determined by the totalization. One is the homological degree, the other will be called the ``internal'' grading. The generalized super Lie algebra $\fn$ is concentrated in homological degree zero, and sits in internal degrees one and two. When we refer to the bidegree, we will mean the ordered pair consisting of homological and internal degree.

Identifying the fiber of the cotangent bundle on~$N$ with $\fn^\vee$, we obtain a decomposition of $\wedge^\bu T^* N$ as a sum of bundles, with respect to the bigrading by homological and internal degree. 
In coordinates, the filtration is easy to see explicitly after writing the de Rham complex of $N$ using a left-invariant, rather than a coordinate, frame. This corresponds to working with the supersymmetry-invariant one-forms on superspace, which take the form
\deq[eq:linv1]{
            \lambda^a = \d\theta^a, \quad v^\mu = \d x^\mu + \theta^a \gamma_{ab}^\mu \d\theta^b.
        }
        The bigrading amounts to assigning $\lambda$ bidegree $(1,-1)$ and $v$ bidegree $(1,-2)$.

        It will often be convenient to use another basis for the bigrading, which totalizes the homological and internal degrees to obtain a consistent $\Z$-grading. We will use this totalized degree together with the negative of the internal degree, which we will refer to as ``weight.'' In other words,
        \[
            (\text{totalized}, \text{weight}) = (\text{homological} + \text{internal}, - \text{internal}).
        \]
        (This basis was called the ``Tate bidegree'' in~\cite{MSJI}.)
        These conventions agree with those used in~\cite[Remark 3.7]{ACDolbeault} in defining the shifted Hodge filtration, up to reversing the order of the Tate bidegree. We explain this further in the next section.
        \numpar[hodge][A Hodge-like filtration]
        For a non-integrable almost complex structure, there is an analogue of the Hodge filtration (and correspondingly of the Dolbeault complex), developed in recent work of Cirici and Wilson~\cite{ACDolbeault}. The de Rham complex is equipped with the structure of a bigraded multicomplex by assigning bidegree $(0,1)$ to $\d \zbar$ and bidegree $(-1,2)$ to $\d z$. Considering the row filtration of this multicomplex defines a spectral sequence whose first differential encodes the torsion of the distribution (in that case, the Nijenhuis tensor). Dolbeault cohomology is defined to be the $E_1$ page of this spectral sequence, which agrees with the standard definition (and with $E_0$) when the distribution is integrable.

        As shown in~\cite{CY2}, the construction generalizes straightforwardly to superspace, with $\lambda$ playing the role of $\d \zbar$ and $v$ playing the role of $\d z$. We prefer to think of a $D_\infty$-algebra structure in cochain complexes, as explained in~\cite{Lapin} (see also~\cite{Boardman,Hurtubise}, as well as~\cite{LWZ} for a recent treatment tailored to the context of Cirici and Wilson's work). 
        Working in coordinates in the left-invariant frame defined above, the de Rham differential on superspace takes the form
        \begin{equation} \label{eq:d-dR}
            \begin{aligned}
                \d_{dR} &= \d_1 + \d_0 + \d_{-1} \\ 
               &= \lambda^a \gamma_{ab}^\mu \lambda^b \pdv{ }{v^\mu} + \lambda^a \left(  \pdv{ }{\theta^a} - \gamma_{ab}^\mu\theta^b \pdv{ }{x^\mu} \right) + v^\mu \pdv{ }{x^\mu}.
           \end{aligned}
        \end{equation}
        It is clear by inspection that the term $\d_m$ has Tate bidegree $(m,1-m)$. We regard the collection $\{\d_m:m\leq 0\}$ as defining a $D_\infty$ structure in chain complexes, with respect to the internal differential $\d_1$. Furthermore, each term in the differential is compatible with the module structure for the supersymmetry algebra, acting by left translations.

        By a familiar abuse of terminology (\S\ref{p:analogy}), we will refer to this spectral sequence as the \emph{Frölicher spectral sequence (for superspace)} in the sequel. In the analogy with the Dolbeault complex, our primary object of study is the $E_1$ page of the Frölicher spectral sequence, which is given by 
\deq{
 W^{\bu,\bu} :=   H^\bu(\Omega^{\bu,\bu}(N), \d_1). 
}
The reader should think of $W^{-k,\bu}$ as playing the role of the complex $\Omega^{k,\bu}$ of Dolbeault forms on a complex manifold, and the differential on the $E_1$ page (which is induced by~$\d_0$) as corresponding to $\dbar$.
This construction recovers the pure spinor superfield formalism and a host of familiar supermultiplets, while giving a fruitful interpretation in terms of the ``almost-complex geometry'' of superspace.
We unpack this connection further in the following sections.
    
\numpar[bundles][The nilpotence variety; multiplets from sheaves]
One feature of the above construction, which may already be apparent to the reader, is that standard constructions in the pure spinor formalism emerge naturally. We view the discussion above as giving a further natural explanation of the origin  of the pure spinor formalism; this explanation is perhaps the most closely connected to superspace geometry. We quickly remind the reader of a few details; for more on our approach to the formalism, we refer to~\cite{perspectives,derivedps}.

With respect to the Tate bidegree, the Chevalley--Eilenberg cochain complex of $\lie{n}$ is a consistently $\Z$-graded cdga with an additional grading by weight.
It is generated by elements $\lambda^a$ and $v^\mu$, in Tate bidegrees $(0,1)$ and $(-1,2)$, respectively. The Chevalley--Eilenberg differential takes the form 
\deq{
    \d_1 = \lambda^a \gamma^\mu_{ab} \lambda^b \pdv{ }{v^\mu}.
}
Neither the notation, nor the similarity to the first term of~\eqref{eq:d-dR}, is a coincidence: the Chevalley--Eilenberg complex $C^\bu(\fn)$ is, as always, isomorphic to the complex of left-invariant forms $\Omega^\bu(N)^N$ on superspace.

The zeroth cohomology of the supertranslation algebra $\fn$ is the quotient of 
the ring $R = \C[\lambda^a]$, equipped with its natural weight grading,
    by the homogeneous ideal $I$ generated by the quadratic equations $\lambda^a \gamma_{ab}^\mu \lambda^b$. 
    We will think of the spectrum of this ring
    as an affine scheme, called the \emph{nilpotence variety} (or generalized pure spinor space, or Maurer-Cartan set) $Y$ of $\fn$. The points of~$Y$ classify possible twists of a theory with $\fn$-supersymmetry; for more information, see~\cite{NV}.

    By tensoring $H^0(\fn) = R/I$ over $\C$ with functions on superspace, we obtain precisely the standard ``pure spinor superfield'' considered in the literature; it appears above as $W^{0,\bu}$, and the standard differential is just $\d_0$. 
    One observes that tensoring $W^{0,\bu}$ over $R/I$ with any equivariant sheaf on $Y$ returns a multiplet~\cite{perspectives,derivedps}; indeed, the procedure defines a functor 
	\begin{equation}\label{eqn:Afunctor}
		A^\bullet_{R/I} : \Mod_{R/I}^{\fg_0} \longrightarrow \Mult_{\fg} 
	\end{equation}
        from equivariant sheaves to $\fg$-multiplets.
        Each of the summands $W^{-k,\bu}$ of the $E_1$ page can thus be understood as the multiplet associated to the equivariant sheaf $H^{-k}(\fn)$.
        One can generalize the formalism to a derived version, in which $C^\bu(\fn)$ plays the role of~$Y$; this then provides a universal superfield formalism in the form of an equivalence of categories between $\fn$-multiplets and $C^\bu(\fn)$-modules~\cite{derivedps}. We will not return to this in the sequel, as we will be primarily interested in the subcategory of multiplets that arise from strict $H^0(\fn)$-modules. With one or two notable exceptions, the standard examples in physics belong to this class.

        It is clear from the above story that the multiplet $A^\bu := W^{0,\bu}$ will 
        play a special role. (By a slight abuse of notation, $A^\bu = A^\bu_{R/I}(R/I)$.) We will discuss this further in the next section, but want to emphasize two important points. The first is that, unlike  a generic multiplet, $A^\bu$ is naturally equipped with a suitably graded \emph{commutative} algebra structure. This structure plays an essential role, both in the construction of Yang--Mills theories and in the construction of first-quantized models; it has been insufficiently appreciated in the past (but see the important work~\cite{MS1,MS2,Mov1,Mov2,GS,GKR,GGST}).
This multiplet, which is canonically determined by specifying the supertranslation algebra $\fn$, was called the \emph{canonical multiplet} in~\cite{MSJI} and the \emph{tautological filtered cdgsa} in~\cite{spinortwist}. 

Secondly, 
multiplets in the essential image of the  functor $A^\bu_{R/I}$---those that correspond to $H^0(\fn)$-modules---acquire an additional structure: they naturally define sheaves of $A^\bu$-modules on superspace.
This additional structure lets us think of such multiplets as ``equivariant quasi-coherent sheaves'' or ``natural holomorphic vector bundles'' on superspace, whereas more general multiplets---while equivariant---fail to be sheaves of $A^\bu$-modules.

    \numpar[geom][The structure sheaf of superspace]
    We have constructed an analogue of the (regraded) Fr\"olicher spectral sequence associated to a superspace based on~$\fn$, and have showed that the $E_1$ page of this spectral sequence is constructed in precisely the manner that Cirici and Wilson construct the Dolbeault complex of an almost-complex manifold. 
    It is extraordinarily profitable to take this analogy seriously, using it as the basic tool for understanding superspace geometry.

    By doing this, as we have emphasized throughout, one is led to equip superspace with \emph{a new structure sheaf}, which is a sheaf of commutative dg superalgebras. In our analogy, the cdgsa $A^\bu =( W^{0,\bu}, \d_0)$ plays the role of $(0,\bu)$-forms equipped with the $\dbar$ differential. These form an appropriate derived replacement for the sheaf of holomorphic functions, and make sense (thanks to Cirici and Wilson) in the almost-complex setting, independent of any integrability requirement. We will thus think of $A^\bu$ as the structure sheaf of superspace; in particular, flat superspace $N$ is, for us, the derived stack $\Spec A^\bu$.
Examples of the supermultiplets $A^\bu$ are listed in Table~\ref{tab:O}; for more details, we refer again to~\cite{perspectives}.

\begin{table}
\begin{tabular}{|l|l|l|l|l|}
\hline
\emph{Superspace} & \emph{Structure sheaf} & \emph{On-shell?}& \emph{$\hdim(N)$} & \emph{CY?} \\ \hline
3d $\N=1$ & BRST vector & & 1 & \\ \hline
4d $\N=1$ & BRST vector & & 2 & \\ \hline
4d $\N=2$ & BRST tensor (dual hyper) & & 1 & \\ \hline
4d $\N=4$ & trivial & & 0 & $\checkmark$ \\ \hline
6d $\N=(1,0)$ & BRST vector & & 3 & \\ \hline
6d $\N=(2,0)$ & presymplectic BV abelian tensor & $\checkmark$ & 1 & \\ \hline
10d $\N=(1,0)$ & BV vector & $\checkmark$ & 5 & $\checkmark$ \\ \hline
10d $\N=(2,0)$ & presymplectic BV IIB supergravity & $\checkmark$ & 1 &  \\ \hline
11d $\N=1$ & BV supergravity & $\checkmark$ & 2 & $\checkmark$ \\ \hline
\end{tabular}
    \caption{Canonical multiplets in some physical examples}
    \label{tab:O}
\end{table}

Using intuitions from (almost) complex geometry as a guide, further comparisons are immediate. 
For example, the largest integer $n$ such that $W^{-n,\bu} \neq 0$ corresponds to the ``dimension'' of superspace. Just as the complex dimension of a complex manifold is not identical to its real dimension as a smooth space, the dimension of superspace is distinct from both the real dimension $d$ and the superdimension $d|k$. 
Since the transferred $D_\infty$ structure witnesses $W^{\bu,\bu}$ as a resolution of the ground field in $A^\bu$-modules, we see that this notion of dimension agrees with the homological dimension of (the local completion of) $A^\bu$; we will call this notion of dimension the \emph{homological dimension} in the sequel.
Its value is given by
 \deq[eq:dimDol]{
        \hdim(N) = \dim(\fn_2) - \dim(\fn_1) + \dim(Y) = d - k + \dim(Y).
 }
 The homological dimension thus measures the failure of $Y$ to be a complete intersection; for complete intersections, $\hdim(N)$ is zero. 

 In~\cite{CY2},
 it was argued that the homological dimension is invariant under twisting (see~\S\ref{twists} below), and is thus equal to the number of surviving translations in the maximal twist of $\lie{n}$. The maximal twist is here defined to be the twist with respect to a supercharge on a stratum of maximal dimension in the nilpotence variety; it is \emph{not} necessarily the twist with the largest number of exact translations.

Further,
the existence of an isomorphism
\begin{equation}
W^{-n,\bu} \cong W^{0,\bu} 
\end{equation}
is analogous to saying that the sheaf of holomorphic top forms can be trivialized over the structure sheaf. This corresponds to a Calabi--Yau structure on superspace. 
Superspace admits a Calabi--Yau structure precisely when the nilpotence variety is Gorenstein~\cite{perspectives,CY2}.
We note that there are subtleties related to duality; in particular, when the module $H^{-i}(\fn)$ is not Cohen--Macaulay, then the multiplets associated to $H^{-i}(\fn)$ and $H^{i-n}(\fn)$ are not dual in the category of multiplets. For details, we refer to~\cite{perspectives}.

\numpar[mu][Relation to component-field models]
By studying multiplets using the pure spinor superfield formalism, we produce models that are freely resolved over smooth functions on superspace---in physical terms, everything is presented in terms of unconstrained superfields. It is natural to ask about how to compare to more familiar component-field models, and in particular to ask which multiplets appear as canonical multiplets.

Each multiplet admits a canonical component-field model, characterized by the conditions that it is described by a chain complex of vector bundles on the underlying bosonic spacetime in which all differentials are differential operators of strictly positive order. To reproduce this model, we can filter the $E_1$ page of the Frölicher spectral sequence (which is $W^{\bu,\bu}$ equipped with differential $\d_0$), again giving rise to the structure of a bigraded multicomplex whose associated spectral sequence abuts to the $E_2$ page of the Frölicher spectral sequence (and thus to the sum of the on-shell multiplets associated to the cohomology groups of~$\fn$). A similar component-field spectral sequence can be defined for the pure spinor superfield description of any multiplet.

Our conventions for the component-field spectral sequence are as follows: we introduce a new degree, termed \emph{$\theta$-degree}, with respect to which all odd generators of the smooth functions on~$N$ carry degree $-1$. This separates the differential into terms
\deq{
    \d_0 = \d_0^{(1)} + \d_0^{(-1)} = \lambda^a \pdv{ }{\theta^a} - \lambda^a \theta^b \gamma^\mu_{ab} \pdv{ }{x^\mu}.
}
In comparison with our conventions for the Frölicher spectral sequence, the weight grading here plays the role of the homological degree, and the $\theta$-degree plays the role of the totalized degree.
Examining the form of $\d_0^{(1)}$ makes it clear that the $E_1$ page of the component-field spectral sequence is computed by considering the Koszul homology of the $R/I$-module underlying the multiplet in question, and can thus be studied  computationally  effectively by looking at the minimal free resolution of that module over~$R$. The transferred $D_\infty$ structure on the $E_1$ page returns the differential on the multiplet, with terms of $\theta$-degree $1-2k$ being differential operators of order $k$. 

\begin{dfn}
For a multiplet $E$, we use the notation $\mu E$ to refer to the component-field model.
This is the $E_1$-page of the spectral sequence of the filtration induced by $\theta$-degree, equipped with the transferred differential.
\end{dfn}
        
  	\numpar[twists][Behavior under twisting]
  	Choosing a square zero element $Q \in Y$ we can ``twist'' the algebra $\fg$ itself by deforming it.
	The twist of $\lie{g}$ with respect to $Q$ is the dg Lie algebra $(\fg , [Q,-])$. We denote the cohomology (which is again a graded Lie algebra in degrees zero, one and two) by
  	\begin{equation}
  	\fg_Q = H^\bullet(\fg , [Q,-])
  	\end{equation}
  	and its strictly positively graded piece by
  	\begin{equation}
  	\fn_Q = H^{>0}(\fg , [Q,-]) ,
  	\end{equation}
as well as the corresponding nilpotence variety by $Y_Q$. In typical examples we will refer to these as the twisted super Poincar\'e, twisted supertranslation algebra and twisted nilpotence variety respectively. For interesting recent work studying higher $L_\infty$ operations on~$\fg_Q$, we refer to~\cite{JKY}.

As explicated in~\cite{spinortwist}, this very naive notion of `twisting' is closely related to the standard procedure of twisting supersymmetric multiplets~\cite{CostelloHolomorphic,ESW}.
We can think of the algebras $\lie{g}_Q$ and $\lie{n}_Q$ as residual symmetry algebras in the twisted theories, and the twisted nilpotence variety $Y_Q$ as the moduli space of further twists.  	
  	The pure spinor superfield formalism can be applied to both $\fg$ and all of its twisted versions; applying it to $\fg_Q$ yields a $\fg_Q$-multiplet. On the other hand, any $\fg$-multiplet can be twisted to give an $\fg_Q$-multiplet. The results of~\cite{spinortwist} indicate that the pure spinor construction is compatible with twisting, such that
\deq[eq:AQAQ]{
            A^\bullet_{\cO_Y}(\cO_Y)^Q \cong A^\bullet_{\cO_{Y_Q}}(\cO_{Y_Q}).
}
        This idea gives rise, not only to enormously powerful computational techniques, but to deep structural connections between a theory and its twists; the essential insight is that, when properly formulated in terms of superspace geometry, \emph{a theory and its twists are alike}.

\subsection{The local superconformal algebra as vector fields on superspace}
Having understood the canonical multiplet $A^\bu$ as analogous to (a Dolbeault resolution of) ``holomorphic functions'' on superspace, and having constructed the multiplets $W^{-k,\bu}$ that play the role of the sheaves of holomorphic differential forms, it is logical to ask what the analogues of other natural sheaves on superspace look like. In this section, we study the derivations of the cdga $A^\bu$, which return the tangent sheaf of our ringed space. We view the resulting multiplet $\Der(A^\bu)$ as analogous to (a derived model of) distribution-compatible vector fields.

\label{sec: t-sheaf}

\numpar[der-A][The tangent sheaf for superspaces]
        Recall that the canonical multiplet is equipped with the structure of a cdgsa and that its derivations thereby naturally form a dg (super) Lie algebra. 
This multiplet of derivations of $A^\bu$ is naturally equipped with a Lie bracket; in fact, with its natural differential it carries the strucutre of a dg super Lie algebra.
Multiplets naturally form vector bundles on superspace (though the rank of the bundle may be infinite). 
Furthermore, the differential and bracket on derivations of $A^\bu$ are given by differential and bidifferential operators respectively.
Thus this multiplet carries the structure of a \textit{local} dg super Lie algebra as defined in~\cite[Definition 6.2.1]{CG1}.

We sum up these considerations with the following definition.
\begin{dfn}
The \textit{local superconformal algebra} $\Conf(\lie{n})$ of type~$\lie{n}$ is the tangent sheaf of superspace: the local super Lie algebra of derivations $\Der(A^\bu)$. It defines a local super Lie algebra on any superspace based on~$\fn$.
    \label{dfn:Conf(n)}
\end{dfn}
In what follows, we will not notationally distinguish between the corresponding sheaf of dg super Lie algebras on superspace and its global sections on a fixed superspace (which carries the structure of a dg super Lie algebra).

The remainder of this paper will go on to characterize $\Conf(\fn)$ more explicitly in various different ways.
We will see that its structure can be thought of as capturing the \emph{full} nonlinear structure of the BRST conformal supergravity multiplet, including all gauge symmetries, and that its field content precisely matches that of the conformal supergravity multiplet in all physical examples. From our perspective, which is informed by derived deformation theory, it is natural to think of $\Conf(\fn)$ as having a moduli-theoretic interpretation, and to see deformations of conformal structure appearing together with the gravitino and other physical fields in degree one.

It is clear that $\Conf(\fn)$ acts on any natural vector bundle on superspace, in the same way that holomorphic vector fields couple to natural holomorphic vector bundles. Thus this multiplet couples naturally to a wide class of other multiplets, including most physically important examples.
We further emphasize that \emph{none} of this structure is constructed in an \emph{ad hoc} or example-dependent manner. The only input datum is the symbol algebra $\fn$, and every further construction is dictated by general considerations.

\numpar[p:derived][The tangent complex]
We remark that the dg Lie algebra of derivations of $A^\bu$ is not the most canonical object to consider from the point of view of deformation theory.
The Koszul duality between commutative algebras and Lie algebras states that formal deformations of the dg super algebra $A^\bu$ are controlled by the dg Lie algebra called its (shifted) \textit{tangent complex} $\mathbb{T}_{A^\bu}$~\cite{SchlSta,Hinich}.
Explicitly, this tangent complex is computed by the so-called Andr\'e--Quillen cohomology of~$A^\bu$~\cite{Hinich,BlockLaz}. As such, the most appropriate definition from the deformation theoretic perspective is the following.
\begin{dfn}
   The \emph{derived superconformal algebra} of type $\fn$ is the tangent complex $\mathbb{T}_{A^\bu}$ to superspace. Again, it defines a local super Lie algebra on any superspace based on~$\fn$.
    \label{dfn:T_A(n)}
\end{dfn}
In many of the examples related to supersymmetry, $\Conf(\lie{n})$ agrees with the multiplet obtained from the tangent complex.\footnote{In fact, we expect that in these examples the dg algebra $A^\bu$ is indeed equivalent to a cofibrant dg algebra.} Moreover, $\Conf(\fn)$ precisely reproduces physical constructions of conformal supergravity multiplets.
Nevertheless, the tangent complex has better-behaved properties in homotopy theory, and should play some natural role in supergravity. It would be interesting to see whether $\mathbb{T}_{A}$ leads, in some example, to any further ``enhancement'' of the local superconformal algebras we study here.

	\numpar[coker][The sheaf of surviving translations]
    We have seen above (\S\ref{bundles}) that non-derived multiplets can be thought of in terms of the sheaves on the nilpotence variety $Y$ that generate them. It turns out that the multiplet $\Conf(\fn)$ is generated by a simple sheaf $\Gamma$ on~$Y$. Thinking of the points of $Y$ as possible twists of a theory with $\fn$-supersymmetry, the stalk of $\Gamma$ at~$Q \in Y$ consists of \emph{those bosonic spacetime translations that survive in the $Q$-twist}. We sum this up in the following theorem.
    Here we make use of the notation $\lie{n}_1 = \Sigma$, $\lie{n}_2 = V$ introduced above in~\S\ref{sec:stdST}.
   
    Given the map $\gamma: \Sym^2(\Sigma) \to V$, we obtain maps $\Sigma \to \Sigma^\vee \otimes V$ and $V^\vee \to \Sym^2(\Sigma^\vee)$. We can view the latter as defining a quadratic-coefficient map of $R$-modules 
    \deq[eq:gamma-hat]{
        \hat{\gamma}: V^\vee \otimes R \to R,
    }
    whose image is $I$. (This map is nothing other than the Chevalley--Eilenberg differential acting on generators dual to $\fn_2$.)
    Similarly, tensoring with $R/I$, we construct the linear-coefficient maps of free $R/I$-modules
            \deq[eq:lambdagamma]{
                \varphi = (\lambda \gamma)^\mu_b: \Sigma \otimes R/I \to V \otimes R/I, 
                \qquad
                \varphi^t : V^\vee \otimes R/I \to \Sigma^\vee \otimes R/I.
            }
            (These are the Jacobian matrix of~$I$ and its transpose.)            

    \begin{thm} \label{thm: coker}
        There is an equivalence of $\lie{n}$-multiplets 
            \deq{ \Conf(\fn) \cong A^\bu_{R/I}(\coker \varphi).
            }
        \end{thm}
        Note that the support of the sheaf $\coker\varphi$, if $\varphi$ is viewed as a map of $R$-modules, does not necessarily lie within the support of the sheaf $R/I$. This means that the use of coefficients $R/I$ in~\eqref{eq:lambdagamma} is essential. 
        \begin{proof}
            We recall that $\Der A^\bu$ consists of the derivations of the underlying graded commutative algebra $(A^\bu)^\#$, equipped with the differential $[\d_0,-]$. We thus begin by giving an explicit description of the underlying graded Lie algebra $\Der (A^\bu)^\#$. 
            Since $(A^\bu)^\# = C^\infty(N) \otimes R/I$, we have that
\deq{
    \Der (A^\bu)^\# = \left( \Vect(N) \otimes R/I \right) \oplus \left( C^\infty(N) \otimes \Der(R/I) \right).
}
Since $C^\infty(N)$ is a smooth algebra, $\Vect(N) = C^\infty(N) \otimes \fn$; we can use the left-invariant frame adapted to the distribution $D$. 
Derivations of~$R/I$ are given by the subspace of elements 
\deq{
    \left\{    f^a(\lambda) \pdv{ }{\lambda^a} \in R/I \otimes \Sigma : f^a(\lambda) \gamma_{ab}^\mu \lambda^b = 0 \right\},
}
which we identify as the $R/I$-module $\ker(\varphi)$.
We can now present $\Der(A^\bu)$ in the form
\begin{equation}
        		\begin{tikzcd}
        		C^\infty(N) \otimes R/I \otimes V &
                        C^\infty(N) \otimes R/I \otimes \Sigma \arrow[l, "\varphi"'] &  
                        C^\infty(N) \otimes \Der(R/I) \arrow[l, hook'] ,
        		\end{tikzcd}
                        \label{eq:threestep}
        	\end{equation}
                where $\Sigma$ denotes sections of the distribution $D$. It is clear by inspection that the cohomology consists of $C^\infty(N) \otimes \coker\varphi$ at the leftmost end; the identification of $\Der(R/I)$ with $\ker\varphi$ ensures that the cohomology in the middle term vanishes, and the right-hand map is injective.
                The remaining terms in the differential just return the internal differential  on $A^\bu_{R/I}(\coker\varphi)$, making the computation into an equivalence of multiplets.
        \end{proof}

        \subsection{Other natural sheaves on superspace} 
       
        \numpar[p:Kaehler][The module of K\"ahler differentials]
        Having characterized the tangent sheaf of superspace, it makes sense to consider the sheaf $\Omega_{A^\bu/\C}$ of K\"ahler differentials.

    \begin{thm} \label{thm: kaehler}
        There is an equivalence of $\lie{n}$-multiplets
        \deq{ \Omega_{A^\bu/\C} \cong A^\bu_{R/I}(\ker \varphi^t).
            }
        \end{thm}
        In analogy to the previous section, we remark that we can think of $\ker\varphi^t$ as the natural sheaf on~$Y$ whose fiber at $Q\in Y$ consists of those constant one-forms on the bosonic spacetime that vanish on all $Q$-exact translations.
        \begin{proof}
            The proof is entirely analogous (or dual) to Theorem~\ref{thm: coker} above. 
            We identify the $R/I$-module of K\"ahler differentials $\Omega_{Y/\C}$ with $\coker(\varphi^t)$, and note that 
            \deq{
                \Omega_{A^\bu/\C} = \left( A^\bu \otimes_\C \fn^\vee \right)\oplus \left( A^\bu \otimes_{R/I} \coker(\varphi^t) \right).
            }
            Taking the differential into account, we can present $\Omega_{A^\bu/\C}$ in the form
            \begin{equation}
		\begin{tikzcd}
                    C^\infty(N) \otimes R/I \otimes V^\vee \arrow[r, "\varphi^t"] &
                    C^\infty(N) \otimes R/I \otimes \Sigma^\vee \arrow[r, two heads] &  
                    C^\infty(N) \otimes \Omega_{(R/I)/\C}.
                    \end{tikzcd}
                    \label{eq:co-threestep}
                \end{equation}
                The cohomology is clearly identified with $A^\bu_{R/I}(\ker \varphi^t)$, and the remaining terms in the differential again return the pure spinor differential. 
        \end{proof}

    \numpar[p: oneforms][One-forms on superspace]
    The attentive reader will have noticed that we have now given \emph{two} plausible definitions of the sheaf of one-forms on superspace. On general grounds, the sheaf of K\"ahler differentials of $A^\bu$ is supposed to play the role of one-forms on~$\Spec A^\bu$. On the other hand, we set up our analogy using the Fr\"olicher spectral sequence for superspace, according to which $W^{-1,\bu} = A^\bu_{R/I}(H^{-1}(\fn))$ is supposed to correspond to the sheaf of structure-preserving one-forms.

    It is natural to ask whether or not these two notions agree, and to look for conditions under which they agree and comparisons between them when they do not. In fact, the relation of the first Koszul homology group to the syzygies of the conormal module has been a subject of great interest in the literature. We will refer primarily to work of Simis and Vasconcelos~\cite{ConormalSyz}, which contains the necessary results for our purposes, but the reader is further referred to references therein.

    For the reader's convenience, we will recapitulate the discussion of the first pages of~\cite{ConormalSyz} here, with notation adapted to our conventions.  
    For brevity, let us write $H$ for the $R/I$-module $\ker(\varphi^t)$; 
    we will further write $Z$ and $B$ for the $R$-modules of cocycles (respectively, coboundaries) of degree $-1$ in~$C^\bu(\fn)$.

    There is an obvious short exact sequence of $R/I$-modules that takes the form
            \begin{equation}
                \begin{tikzcd}
                    0 \ar[r] &  
                    H \ar[r] &  V^\vee \otimes R/I \ar[r, "\varphi^t"] & \im \varphi^t 
                    \ar[r] & 0.
                \end{tikzcd}
                \label{eq: conormal}
            \end{equation}
            To understand the relation of~$H$ to~$H^{-1}(\fn) = Z/B$,  
consider the short exact sequence of $R$-modules arising from~\eqref{eq:gamma-hat}:
            \begin{equation}
                \begin{tikzcd}
                    0 \ar[r] &  Z \ar[r] &  V^\vee \otimes R \ar[r, "\hat\gamma"] & I \ar[r] & 0.
                \end{tikzcd}
                \label{eq: Z-SES}
            \end{equation}
            (By definition, $Z = \ker(\hat\gamma)$.) By tensoring this sequence with $R/I$, we get a four-term exact sequence of $R/I$-modules which takes the form
\begin{equation}
                \begin{tikzcd}
                    0 \ar[r] & \Tor_1^R(I,R/I) \ar[r] &   Z/IZ \ar[r] &  V^\vee \otimes R/I \ar[r, "\varphi^t"] & I/I^2\ar[r] & 0.
                \end{tikzcd}
                \label{eq: tor1}
            \end{equation}
            Comparing~\eqref{eq: tor1} to~\eqref{eq: conormal}, we note that we can identify $\im(\varphi^t)$ with the conormal module $I/I^2$. Furthermore, we can identify $\Tor_1^R(I,R/I)$ explicitly with $\left( Z \cap (V^\vee \otimes I) \right)/IZ$, the intersection taking place inside $V^\vee \otimes R$. Thus $H = Z/ Z \cap (V^\vee \otimes I) $.

            The chain of inclusions $B \subset  Z \cap (V^\vee \otimes I) \subset Z$ (where the first inclusion follows from inspecting the form of the Chevalley--Eilenberg differential) gives rise to another short exact sequence of $R/I$-modules,
\begin{equation}
                \begin{tikzcd}
                    0 \ar[r] & Z \cap (V^\vee \otimes I) / B  \ar[d, equal, "\,\text{def.}"] \ar[r] &  Z/B \ar[r] \ar[d, equal] &  Z/ Z \cap (V^\vee \otimes I) \ar[r] \ar[d, equal] & 0 \\
                    0 \ar[r] & \delta(I) \ar[r] &  H^{-1}(\fn)  \ar[r] &  H \ar[r] & 0, 
                \end{tikzcd}
                \label{eq: compare}
            \end{equation}
            witnessing the comparison between $H = \ker(\varphi^t)$ and the first Koszul homology.
            Simis and Vasconcelos refer to $H$ as the \emph{syzygy part} of the first Koszul homology; they go on to show that $\delta(I)$, as defined in the above diagram, is equivalent to $\ker(\Sym^2(I) \to I^2)$ (with the map being the obvious one defined by multiplication). An ideal is termed \emph{syzygetic} when $\delta(I)$ vanishes; the two notions of one-forms on superspace agree precisely when $I$ is syzygetic.

            For us, given our focus on the Fr\"olicher spectral sequence for superspace, it is most natural to think of $W^{-1,\bu}$ as the sheaf of one-forms on $\Spec A^\bu$. What we have said in this section proves that there is a canonical map 
            \deq{
                W^{-1,\bu} \to \Omega_{A^\bu/\C},
            }
            induced by applying the pure spinor functor to the surjection $H^{-1}(\fn) \to H$ of $R/I$-modules. Under the pure spinor correspondence, the sheaf $W^{-1,\bu}$ corresponds to the first Koszul homology of~$I$; the sheaf of K\"ahler differentials $\Omega_{A^\bu/\C}$ corresponds to the \emph{syzygy part} of the first Koszul homology.

\section{Universal and exceptional cohomology}
\label{s:universal}

\subsection{Superconformal algebras and prolongations}
In this section, 
we will unpack the structure of the local superconformal algebra in a bit more detail. 
We begin by characterizing the zeroth cohomology of~$\Conf(\fn)$; in combination with the results of~\cite{AltomaniSanti}, this shows how 
Nahm's list of superconformal algebras appears in our context.

\numpar[p:H0][Cohomology in degree zero]
We have the following general characterization:
\begin{prop}
Let $\fn$ be a supertranslation algebra. The global sections of $H^0(\Conf (\fn))$ on the corresponding flat model are isomorphic to the maximal transitive prolongation of~$\fn$. When $\fn$ is a physical supertranslation algebra, $H^0(\Conf(\fn))$ on flat superspace reproduces the 
superconformal algebra whenever it exists, and the super Poincar\'e algebra extended by scale transformations in all other instances.
\label{prop:prolong}
\end{prop}

\begin{proof}
Recall that for any superconformal structure $(M, D)$, $\Conf(\fn)$ is a model for the sheaf of vector fields preserving the distribution $D$. By definition, the zeroth cohomology of its sections on $(M,D)$ describe the infinitesimal symmetries of $(M, D)$.

In~\cite{KST}, the infinitesimal symmetries of an arbitrary superconformal structure are described in terms of the maximal transitive prolongation of the underlying supertranslation algebra $\fn$. In particular, in~\S4.2, it is shown that the infinitesimal symmetries of the flat model $N$ coincide with the maximal transitive prolongation of~$\fn$. 

Altomani and Santi compute the maximal transitive prolongations for the standard list of ordinary supertranslation algebras and find the following results~\cite[Theorem 5.1]{AltomaniSanti}:
\begin{itemize}[label={---}]
	\item In dimensions one  and two, the maximal transitive prolongations of supertranslation algebras are infinite-dimensional and isomorphic to the contact super Lie algebra $K(1|\N)$ and to $K(1|\N_L)\oplus K(1|\N_R)$ respectively.
	\item When a finite-dimensional superconformal algebra exists---in dimensions three and four, as well as for $\cN=1$ in five dimensions and $(\cN,0)$ supersymmetry in six dimensions---the maximal transitive prolongation of the supertranslation algebra is isomorphic to the respective superconformal algebra.
	\item In all other cases, the maximal transitive prolongation is isomorphic to the semidirect product
	\deq{
	\fg = \fg_0 \ltimes \fn
	}
	of the supertranslation algebra and its infinitesimal automorphisms. In particular, in degree zero, one finds Lorentz transformations, $R$-symmetry, and a copy of $\mathfrak{gl}(1)$ acting by conformal weight.
\end{itemize}
The claim follows.
\end{proof}

\subsection{Universal component fields in $\Conf(\fn)$}
In this section we describe certain universal cohomology classes in $\mu \Conf(\lie{n})$. These characterize particular component fields that always appear in the corresponding multiplet. The results follow from a computation of a portion of the cohomology that does not depend sensitively on any details of the structure map $\gamma$. The only assumption we make at this point is that $\gamma$ is surjective.
\numpar[grDer][Gradings]
We organize our computation using a variant of the bigrading on the canonical multiplet $A^\bu$ introduced above in~\S\ref{mu}.
Recall that we write $\lambda$ for a coordinate on the nilpotence variety, $\theta$ for the odd spinor coordinate, and $x$ for the spacetime coordinate, and that all indices are left tacit where possible. We will compute $\Conf(\fn)$ by first understanding derivations of $A^\bu$ without reference to the differential, and then equipping them with the adjoint action of~$\d_0 \in \Der(A^\bu)$ by commutator. Our conventions will allow us to make use of the component-field spectral sequence defined in~\S\ref{mu} to understand $\mu \Conf(\fn)$, as they are compatible with the description given in Theorem~\ref{thm: coker}.

The first grading is by weight, which is identified with polynomial degree in $R$ and which plays the role of the cohomological grading on $ \Conf(\fn)$. The differential $\d_0$ is of weight one.
The other grading is by the $\theta$-degree, but we will place the differential operator $\partial/\partial x$ in degree $+2$. Stated differently, we extend the internal grading on the supertranslation algebra to a grading on vector fields on~$N$ via the left-invariant frame; however, we also assign a nontrivial weight to the odd coordinates $\theta$, while leaving smooth functions of~$x$ in degree zero. Having done this, the differential $\d_0$ is of homogeneous $\theta$-degree $+1$.
We list the bidegrees of specific generators in Table~\ref{tab:moredegs}.

\begin{table}
    \begin{tabular}{c|r|r}
        {Generator} & {Weight} & {$\theta$-degree} \\
        \hline
        $x$ & $0$ & $0$ \\
        $\theta$ & $0$ & $-1$ \\
        $\lambda$ & $1$ & $0$ \\
        \hline
        $\partial/\partial{x}$ & $0$ & $+2$ \\
        $\partial/\partial{\theta}$ & $0$ & $+1$ \\
        $\partial/\partial{\lambda}$ & $-1$ & $0$ 
    \end{tabular}
    \caption{Degree conventions for $\mu \Conf(\fn)$}
    \label{tab:moredegs}
\end{table}

With respect to this bigrading on the cochain complex $\Conf(\fn)$, the Lie bracket respects the weight grading (it is homogeneous of degree zero). With respect to the $\theta$-degree, the Lie bracket decomposes into terms of degrees $0$ and $-2$, where the degree-$(-2k)$ terms are $k$-th-order differential operators in the spacetime coordinates $x$. 

Since $\d_0$ is of homogeneous bidegree $(1,1)$, the differential $[\d_0,-]$ on~$\Conf(\fn)$ contains terms of bidegree $(1,1)$ and $(1,-1)$ that are of zeroth and first order in spacetime derivatives respectively. We denote these by $\d_0^{(1)}$ and~$\d_0^{(-1)}$ respectively.
If we apply the component-field spectral sequence using the conventions adopted above, we first take the cohomology of $\d_0^{(1)}$.
The transferred $D_\infty$ structure on $\mu \Conf(\fn)$ will contain terms of bidegree $(1,1-2k)$ that are spacetime differential operators of order $k$, where now $k \geq 1$. 

\numpar[p:complex][Elements of low degree in $\Conf(\fn)$]
We recall the description of $\Conf(\fn)$ given in 
the proof of Theorem~\ref{thm: coker} above. 
We have that 
\deq{
    \Conf(\fn) = \left( A^\bu \otimes_\C \fn\right)  \oplus \left( A^\bu \otimes_{R/I} \Der(R/I) \right),
}
equipped with the differential $[\d_0,-]$. $\fn$ consists of $V$ in bidegree $(0,2)$ and $\Sigma$ in bidegree $(0,1)$; the piece of the algebra $A^\bu$ in bidegree $(i,-j)$ is $C^\infty(V) \otimes \wedge^j \Sigma^\vee \otimes (R/I)^i$.

We represent the low-lying summands in the bigraded decomposition of $\Conf(\fn)$ explicitly in Table~\ref{tab:confs}. Here, ``low-lying'' means that we depict only summands in bidegree $(i,j)$ with $j-i\geq -1$.
Both gradings are in decreasing order; the vertical grading is the weight, and the horizontal grading is the $\theta$-degree.
We suppress the copy of smooth functions on spacetime that occurs throughout, so that we are interested in the bigraded $C^\infty(V)$-module freely generated by the summands in the table.

\begin{table} 
\[
    \begin{tikzcd}[column sep = 1.4ex]
    & \ul{2} & \ul{1} & \ul{0} & \ul{-1} \\
\ul{0} & V & \Vectorstack{{\Sigma} {(\wedge^1 \Sigma^\vee \otimes V)}} \ar[dl] & 
\Vectorstack{{\Der(R/I)^0} {\wedge^1 \Sigma^\vee \otimes \Sigma} {\wedge^2 \Sigma^\vee \otimes V}} \ar[dl] 
&
\Vectorstack{{\wedge^1 \Sigma^\vee \otimes \Der(R/I)^0} {\wedge^2 \Sigma^\vee \otimes \Sigma} {\wedge^3 \Sigma^\vee \otimes V}} \ar[dl] 
\\
\ul{1} &  (R/I)^1 \otimes V & 
\Vectorstack{{(R/I)^1 \otimes \Sigma} {(R/I)^1 \otimes \wedge^1 \Sigma^\vee \otimes V}} \ar[dl]  
&
\Vectorstack{{\Der(R/I)^1} {(R/I)^1 \otimes \wedge^1 \Sigma^\vee \otimes \Sigma} {(R/I)^1 \otimes \wedge^2 \Sigma^\vee \otimes V}} \ar[dl]
\\
\ul{2} & (R/I)^2 \otimes V 
&
\Vectorstack{{(R/I)^2 \otimes \Sigma}
{(R/I)^2 \otimes \wedge^1 \Sigma^\vee \otimes V}}
\ar[dl]
\\ 
\ul{3} & (R/I)^3 \otimes V
\end{tikzcd}
\]
\caption{Low-lying generators of~$\Conf(\fn)$}
\label{tab:confs}
\end{table}

As described above, we are interested in $\mu \Conf(\fn)$, and so pass to the associated graded of the filtration by order of differential operator described in the previous paragraph. The differential $\d_0^{(1)}$ on the associated graded is commutator with $\d_0$---but where $\partial/\partial{x}$ does not act. (For this reason, the associated graded is referred to as ``zero mode cohomology'' in~\cite{Cederwall}.) The differential has bidegree $(1,1)$, and acts down and to the left in the table. The associated graded thus splits as a direct sum of subcomplexes with fixed values of $j-i$. 

In this section, we compute the low-lying cohomology of $\mu \Conf(\fn)$. The Lie algebra $\fg_0$ of derivations of~$\fn$ will play a role here. Recall that, by definition, there is an injective map 
\deq{
    (\rho_1,\rho_2): \fg_0 \to \End(\Sigma) \oplus \End(V);
}
$\rho_1$, respectively $\rho_2$, thus denote the composition of this map with projection on the first or second factor.

\begin{thm}
    Let $\fn$ be a generalized supertranslation algebra for which the structure map $\gamma: \Sym^2 \Sigma \to V$ is surjective. 
    Then the low-lying bi-graded components of $\mu \Conf(\lie{n})$ have the following universal cohomology:
    \begin{enumerate}
        \item The diagonal with $j-i=2$ consists of smooth vector fields on spacetime in bidegree $(0,2)$.
\item The diagonal with $j-i=1$ is is generated over smooth functions by the derivations $Q = \frac{\del}{\del \theta} + \theta \frac{\del}{\del x}$ in bidegree $(0,1)$.
    \item  The diagonal with $j=i = 0$ is quasi-isomorphic to the two-term complex 
        \[
            \begin{tikzcd}
                \fg_0 \ar[r, "\rho_2"] & \End(V)
            \end{tikzcd}
        \]
        in bidegrees 
        $(0,0)$ 
        and~$(1,1)$. Thus the cohomology in bidegree $(0,0)$ is generated over smooth functions by the $R$-symmetry Lie algebra $\ker(\rho_2)$.
    \end{enumerate}
    \label{thm: univ}
    \begin{proof}
        The computation of the cohomology can be done separately for each diagonal subcomplex, labeled by values of $j-i$ between $0$ and $2$. These computations will be performed in the next paragraphs (\S\S\ref{p:two}--\ref{p:zero}). We will use yet another auxiliary spectral sequence, coming from the three-step filtration (of cochain complexes) by the subspaces
        \deq[eq:filt]{
                0 \subset A^\bu \otimes \fn_2 \subset A^\bu \otimes \fn \subset \Conf(\fn).
            }
            (Compare this presentation to that used above in~\eqref{eq:threestep}.)
        We can see that this is a filtration by noting that the left action of $\d_0$ on $\Der(A^\bu)^\#$ just induces the internal differential on~$A^\bu$, which clearly preserves the filtration. The other terms appear when an element of~$\Der(A^\bu)^\#$ acts on the left on~$\d_0$; these act by sending
        \deq{
            \pdv{ }{\lambda^a} \mapsto D_a, \quad \pdv{ }{\theta^a} \mapsto \lambda^b \gamma^\mu_{ab} \pdv{ }{x^\mu}, \quad
            \pdv{ }{x^\mu} \mapsto 0,
        }
        and thus also preserve the filtration~\eqref{eq:filt}. This spectral sequence will collapse to the cohomology of the associated graded of the filtration defining the component-field spectral sequence.
        The result follows from these computations.
   \end{proof}
\end{thm}

        At the $E_1$ page of our auxiliary spectral sequence, we note that we obtain a direct sum of the component fields of the canonical multiplet, tensored with $\fn$, and the component fields of the multiplet associated to the $R/I$-module $\Der(R/I) = \ker(\varphi)$.

\numpar[p:two][The diagonal at $j-i=2$]
The only term, generated in bidegree $(0,2)$, is the vector space $C^\infty(V) \otimes V$, which we identify with infinitesimal changes of coordinates---in other words, smooth vector fields on $M$.
These are included as ghosts in any theory of conformal supergravity.

\numpar[p:one][The diagonal at $j-i=1$]
The subcomplex on this diagonal takes the form
\[
    \begin{tikzcd}[row sep = 1.5 ex] 
        \ul{(1,2)} & \ul{(0,1)} \\
        & \Sigma \ar[ld, "-\gamma^*"'] \\
        \Sigma^\vee \otimes V & \Sigma^\vee \otimes V, \ar[l, "1"]
    \end{tikzcd}
\]
recalling that $(R/I)^1 = \Sigma^\vee$.
Here, we write $\gamma^*$ for the map $\Sigma \to \Sigma^\vee \otimes V$ obtained by dualizing $\gamma$.

Applying the auxiliary spectral sequence, it is immediate already at $E_1$ that the cohomology is isomorphic to~$\Sigma$ in bidegree $(0,1)$. But it is useful for the intuition to work out representatives explicitly here. 
In coordinates, the differential acting in bidegree $(0,1)$ is given by
\deq{
\pdv{ }{\theta^a } \mapsto \lambda^b \gamma_{ab}^i \pdv{ }{x^i}, 
\qquad
\theta^a \pdv{ }{x^i} \mapsto \lambda^a \pdv{ }{x^i}.
}
The contraction not determined by the structure map of the generalized supertranslation algebra is eliminated completely in cohomology; the left-invariant combination $Q = \frac{\del}{\del \theta} + \theta \frac{\del}{\del x}$ survives and spans the cohomology in bidegree $(0,1)$ as a $C^\infty(V)$-module (at the level of $\mu \Conf(\fn)^\#$).
We recognize this as the space of ghost fields for local supersymmetry in conformal supergravity.
(The appearance of $Q$ in cohomology corresponds to the standard fact that $Q$ and $D$ commute.)

\numpar[p:zero][The diagonal at $j-i=0$]
We analyze the next diagonal, starting in bidegree $(0,0)$. 
Noting that the linear component $\Der(R/I)^0$ consists (almost by definition) precisely of degree-zero automorphisms of~$\fn$, 
we see that it takes the form
\begin{equation}
    \begin{tikzcd}[row sep = 1.5 ex]
        \ul{(2,2)} & \ul{(1,1)} & \ul{(0,0)} \\
        & & \fg^0 \ar[dl, "\lambda D"'] \ar[ddl] \\
    & \Sigma^\vee \otimes \Sigma \ar[dl, "\gamma^*"] & \Sigma^\vee \otimes \Sigma \ar[l] \ar[dl,"-\gamma^*"] \\
    (\Sym^2 \Sigma^\vee / \im \gamma) \otimes V & (\Sigma^\vee \otimes \Sigma^\vee) \otimes V \ar[l] & (\wedge^2 \Sigma^\vee) \otimes V \ar[l]
        \end{tikzcd}
    \label{eq:diag0}
\end{equation}

    The $E_1$ page of our spectral sequence is computed by taking cohomology with respect to the terms in the differential that act horizontally. We first look at the subcomplex on the bottom row of~\eqref{eq:diag0}.
    This is a quotient of the Koszul complex; its cohomology is thus concentrated in bidegree $(1,1)$, and is isomorphic to the subspace spanned by the image of the dual of the structure map $\gamma \colon S^2 \Sigma \to V$:
\beqn
\im(\gamma^\vee) \otimes V \subset \Sigma^\vee \otimes \Sigma^\vee \otimes V .
\eeqn
Under the assumption that $\gamma$ is surjective, this cohomology is simply $\End(V)$. 
The middle row is clearly acyclic, so that the $E_1$ page just reduces to the cochain complex 
\[
    \begin{tikzcd}[row sep = 1.5 ex]
        \ul{(1,1)} & \ul{(0,0)} \\
        \End(V)  & \fg^0 \ar[l, "\rho_2"'] .\\
        \end{tikzcd}
    \]
    By definition, the cohomology in bidegree $(0,0)$ is the $R$-symmetry Lie algebra; we recognize ghost fields for local $R$-symmetry transformations in conformal supergravity. The cohomology in bidegree $(1,1)$ is the cokernel of the Lie algebra of spacetime symmetries inside of~$\End(V)$. 

    We remark that, in standard physical examples, $\fg_0$ acts on the spacetime via rescaling and via Lorentz transformations $\wedge^2(V) \cong \lie{so}(V)$, so that the cohomology in bidegree $(1,1)$ consists of traceless symmetric endomorphisms of the tangent bundle.
    What we have said is enough to ensure that these elements are represented in cohomology by the elements 
    \deq[eq:metreps]{
        f = f^{\mu}_i (x)  \left( \lambda^a \gamma_{ab}^i \theta^b \pdv{ }{x^\mu} \right),
    }
    with $f^{\mu}_i$ traceless and symmetric with respect to the background frame. (We will explain this interpretation further below in~\S\ref{sec: Lie}.) 

\subsection{Universal symmetries in degree zero}  \label{ssec: universal-symm}      
Proposition~\ref{prop:prolong} characterizes the zeroth cohomology of~$\Conf(\fn)$ on flat space as the maximal transitive prolongation of~$\fn$. It follows that $\fg = \fg_0 \ltimes \fn$ always appears as a sub Lie algebra of~$H^0(\Conf(\fn))$. In this section, we identify the 
explicit vector fields that constitute this subalgebra.
 This is the ``universal'' piece of $H^0(\Conf(\fn))$; for the standard supertranslation algebras, $H^0(\Conf(\fn))$ is only larger when an exceptional prolongation (and thus a superconformal algebra) exists.

\numpar[p:vecferm][Odd vector fields]
We have seen that local supersymmetry transformations always appear as component fields, spanned by the left-invariant vector fields
 	\begin{equation}
 		Q_a = \frac{\partial}{\partial \theta^a} + \gamma^\mu_{ab} \theta^b \frac{\partial}{\partial x^\mu}.
 	\end{equation}
in bidegree~$(0,1)$. 
A generic section takes the form $\sigma = \sigma^a Q_a$, where $\sigma^a$ are smooth functions on~$N_2$.
Recalling the decomposition $\d_0^{(1)} + \d_0^{(-1)}$ of the differential from~\S\ref{grDer}, we note that $\d_0^{(1)}(\sigma^a Q_a) = 0$.
The differential sends such a generic section to 
\deq[eq:varsigma]{
[\lambda^b D_b, \sigma^a Q_a] =\d_0^{(-1)}(\sigma^a Q_a) =  \left( \gamma_{bc}^\mu \del_\mu \sigma^a \right)  \lambda^b \theta^c Q_a.
}
There are thus always cohomology classes corresponding to the kernel of this map. These are the \emph{global supersymmetries}, for which $\sigma^a$ is constant.  As we will see in what follows, exceptional cohomology classes arise when sections $\sigma^a Q_a$ that are not annihilated by the differential can be corrected by the $\d_0^{(1)}$-image of additional terms in bidegree $(0,-1)$. (In general, a longer zigzag may be necessary, in the manner of homological perturbation theory.)

\numpar[p:vecbos][Even vector fields]
We begin with even smooth vector fields, which always appear as component fields in bidegree~$(0,2)$. A generic section takes the form $X = X^\nu \del_\nu$, and is sent by the differential to
\deq{
[\lambda^b D_b, X^\nu \del_\nu] = \d_0^{(-1)}(X^\nu \del_\nu) =  \left( \gamma_{bc}^\mu \del_\mu X^\nu \right) \lambda^b \theta^c \del_\nu.
}
Again, there are always cohomology classes corresponding to the kernel of this map. These are \emph{global translations}, for which $X^\nu$ is constant.

If this expression is in the image of $\d_0^{(1)}$, we can find a correction term in bidegree ${(0,0)}$. From Theorem~\ref{thm: univ}, it is immediate that this will happen precisely when $\d_0^{(-1)}(X)$---which can be identified as the Jacobian of~$X$, viewed as an element of~$\lie{gl}(V)$ using the background frame---is in the image of~$\rho_2$. (We note that this is exactly the condition of being a conformal Killing vector field, which is necessary---but not sufficient in general---for a correction to exist.)

Linear-coefficient vector fields map to constant expressions in bidegree~$(1,1)$, so that the corresponding correction term also has constant coefficients and is automatically closed with respect to~$\d_0^{(1)}$.
The fact that $\d_0$ is invariant under~$\fg_0$ in its standard action on the generators of~$A^\bu$
 is enough to ensure that the corresponding linear vector fields 
 are closed and generate the copy of $\fg_0$ in cohomology in bidegree $(0,0)$. Explicitly, given an element $g \in \fg_0$, these are
\deq[eq:g0reps]{
		\rho_2(g)_{\mu}^\nu x^\mu \frac{\partial}{\partial x^\nu} + \rho_1(g)_{a}^{b}\left( \theta^a \frac{\partial}{\partial \theta^b} + \lambda^a \frac{\partial}{\partial \lambda^b}
		 \right).
}
The key computation involved in the zigzag is that
\deq{
\d_0^{(1)}\left( \rho_1(g)_{a}^{b}\left( \theta^a \frac{\partial}{\partial \theta^b} + \lambda^a \frac{\partial}{\partial \lambda^b} \right) \right) =  \rho_2(g)_{\mu}^\nu \left( \lambda^c \gamma^\mu_{cd} \theta^c \right) \frac{\partial}{\partial x^\nu} ,
}
using the fact that $\gamma$ is, by definition, a $\fg_0$-equivariant map.

 \subsection{The superconformal algebra in an example}

Whenever superconformal algebras exist in the traditional sense, the universal part of $H^0(\Conf(\fn))$ is extended. Here,  we study an explicit physical example ($\N=1$ supersymmetry in three dimensions), showing how the zeroth cohomology enlarges to recover the superconformal algebra $\lie{osp}(1|2)$. Upon pushing forward to the smooth supermanifold $N$, we find that the supervector fields spanning the zeroth cohomology match with the known conformal Killing supervector fields in this example~\cite{Park3d}.

The advantage of working in this example is that $\Conf(\fn)$ is of finite rank as a vector bundle, even at the cochain level. It is supported in degree zero in bidegrees between $(0,2)$ and~$(0,-2)$, and in degree one between bidegrees $(1,2)$ and~$(1,-1)$. Above degree one, it is trivial. Thus no further homotopy corrections to cohomology classes in bidegree~$(0,1)$ via additional zigzags are possible. This streamlines the computation, as does the fact that any quadratic expression in $\lambda$ automatically vanishes.

We can furthermore simplify the notation in this example by recalling that $V \cong \Sym^2(S)$ as representations of~$\lie{so}(3)$. So we can replace an abstract vector index by a symmetric pair of abstract spinor indices, and we will do this in what follows. We will also sometimes use the notation $(a \wedge b)$ for the antisymmetric $\lie{so}(3)$ invariant bilinear pairing on~$S$.

In addition to the supervector fields described in~\S\ref{ssec: universal-symm}, one finds the following cohomology classes in degree zero:

\numpar[p:Svecs][Conformal supersymmetries] 
As in the second part of~\S\ref{p:vecbos}, we start with a general element $\sigma^a Q_a$ in bidegree $(0,1)$. Using the full computation of the component fields (which is presented below in~\S\ref{3dN=1}), we observe that the spin-$3/2$ portion of the expression in~\eqref{eq:varsigma} must vanish in order for a correction using $\d_0^{(1)}$ to be possible. Thus $\sigma$ must be in the kernel of the \emph{Penrose operator:} it is a conformal Killing spinor, or \emph{twistor spinor}~\cite{BaumTwistor}.

The solutions to this condition have linear coefficients; they take the form 
$x^{ab}Q_b$.
Applying the differential produces the element 
\deq[eq:to-kill]{
\lambda^a \theta^b Q_b + \theta^a \lambda^b Q_b.
}
This element can be nullhomotoped in unique fashion with respect to~$\d_0^{(1)}$. To see this, note the identity
\deq{
\left[ \lambda^b D_b, \lambda^a \theta^c \pdv{ }{\lambda^c} \right] = \lambda^a \theta^c D_c =  \lambda^a \theta^c \pdv{ }{\theta^c} =  \lambda^a \theta^c Q_c,
}
which accounts for the exactness of the first term in~\eqref{eq:to-kill}. 
To nullhomotope the second term, let $E_\theta$ and $E_\lambda$ denote the corresponding Euler vector fields. Observe that
\begin{equation}
\begin{gathered}
\left[ \lambda^b D_b, \theta^a E_\theta \right] = 
\lambda^a E_\theta - \lambda^b \theta^a [D_b, E_\theta] = \lambda^a E_\theta - \theta^a \lambda^b Q_b, \\
\left[ \lambda^b D_b, \theta^a E_\lambda \right] = \theta^a \lambda^b D_b.
\end{gathered}
\end{equation}
Finally, observe that
\deq{
\left[ \lambda^b D_b, (\lambda\wedge\theta) 
\pdv{ }{ \lambda^a} \right]  = \lambda^c \theta^d \eps_{cd} D_a = 
(\lambda\wedge\theta) \pdv{ }{\theta^a} - 
(\lambda\wedge\theta) \theta^b \del_{ab} = 
(\lambda\wedge\theta) \pdv{ }{\theta^a} + \lambda^b \theta^2 \del_{ab}.
}
Recalling the Fierz identity for $\lie{sl}(2)$, which implies in this case that
\deq{
\lambda^a E_\theta - \eps^{ab} (\lambda\wedge\theta) \pdv{ }{\theta^b} = \theta^a \lambda^b \pdv{ }{\theta^b},
}
we see that
\deq{
 \theta^a E_\lambda - \theta^a E_\theta + \eps^{ab} (\lambda\wedge\theta) \pdv{ }{\lambda^b}
}
furnishes the required nullhomotopy of the second term. Since this term has constant coefficients, it is automatically closed for~$\d_0^{(-1)}$, so that we obtain the expression
\deq[eq:confsups]{
S^a = x^{ab} Q_b  - \theta^a E_\theta  + \lambda^a \theta^c \pdv{ }{\lambda^c} + \theta^a E_\lambda + \eps^{ab} (\lambda\wedge\theta) \pdv{ }{\lambda^b}
}
for the degree-zero cohomology classes corresponding to conformal supercharges.
 
 \numpar[p:Kvecs][Special conformal transformations]
 For the even part, we can find the vector fields corresponding to special conformal transformations either by performing a two-step zigzag starting with the usual expression for a quadratic-coefficient conformal Killing vector field, or by evaluating the Lie bracket of two of the odd supervector fields in~\eqref{eq:confsups}. 
The first correction has linear coefficients, and sits along the generators
\deq{
 \left( \theta^a \frac{\partial}{\partial \theta^b} + \lambda^a \frac{\partial}{\partial \lambda^b}		 \right),
}
just as in~\eqref{eq:g0reps}. 
This is no longer annihilated by~$\d_0^{(-1)}$, but a further correction is possible.
The second correction sits along the generators
\deq{
\lambda^a \theta^2 \pdv{ }{\lambda^b}.
}
We omit the details, which are not especially instructive. In total, we see explicitly how the ordinary superconformal algebra $\mathfrak{osp}(1|2)$ associated to three-dimensional $\cN=1$ supersymmetry reemerges in our model in this example.

\numpar[p:pushforward][Pushing forward to smooth superspace]
There is a map from $A^\bu$ to $C^\infty(N)$, analogous to the map of cdga's from the Dolbeault complex to smooth functions. At the level of the underlying commutative superalgebra, it arises from the quotient by the maximal ideal of~$R/I$. This map is compatible with the differential $\d_0$ on~$A^\bu$.

Since $\Conf(\fn)$ is a sheaf of~$A^\bu$-modules, we can use this map to define base change to~$C^\infty(N)$. We think of this as the comparison map that views a distribution-preserving vector field simply as a smooth vector field on superspace.

Base changing the special conformal generators in~\eqref{eq:confsups}, we find a match with the normal conformal Killing supervector fields, as presented in~\cite[\S3.2]{Park3d}. (Since the bosonic special conformal generators arise as Lie brackets of $S$-type generators, it is sufficient to explicitly match there.)

\section{$G$-structures from superconformal structures}
\label{sec: Lie}

\subsection{The local conformal algebra}
\label{ssec:bosonic}

We begin by presenting a local derived model for deformations of conformal classes of metrics; our discussion follows~\cite[\S 12.6]{CG2}, but we refer also to~\cite{KapranovConf}. 
We remark that some related ideas have also appeared in the recent physics literature~\cite{Grigoriev}. We then go on to present a different local model for conformal structures, which makes use of a frame (vielbein), and to prove that the models are equivalent as $L_\infty$ algebras. In fact, both models are strict, but they are related only by an $L_\infty$ morphism, reflecting the fact that the metric is quadratic in the frame. It is the frame model which is most closely related to viewing a conformal structure as a reduction of structure group; this model therefore appears most naturally in the context of deformations of superconformal structures.

\numpar[p:Lconf][A model for the moduli problem of conformal classes of metrics]
We consider a complex of natural vector bundles, defined on any (pseudo-)Riemannian manifold. For later convenience, we primarily make use of the \emph{inverse} Riemannian metric, denoted $g$; this is an invertible map from $\T^*$ to $\T$, or equivalently a section of  the symmetric square bundle $S^2(\T)$ that is globally of full rank. (Since inverse metrics are in one-to-one correspondence with metrics, our model is equivalent to the one given in~\cite{CG2}, but via an $L_\infty$ quasi-isomorphism with infinitely many corrections, corresponding to the terms in the Taylor series expansion of the formal family $(g + th)^{-1}$.) The complex we use can be presented as follows:
\beqn
\begin{tikzcd}
\ul{0} & \ul{1} \\
\T \ar[r,"L"] & S^2(\T) \\
C^\infty \ar[ur,"g"]  &
\end{tikzcd}
\label{eq:Lconf}
\eeqn
where $L (X) = L_X g$ is the Lie derivative of $g$ along $X$ and the diagonal map is $\lambda \mapsto \lambda g$.

\begin{dfn}
    The \emph{local Lie algebra of conformal classes} $\cL_\Weyl$ is the chain complex of vector bundles~\eqref{eq:Lconf}, equipped with the
    (strict) local Lie algebra structure defined by the brackets:
\begin{itemize}
\item $[X,Y] = L_X Y$ where $X,Y$ are vector fields.
\item $[X,\lambda] = X (\lambda)$ where $\lambda\in C^\infty$.
\item $[X,h] = L_X h$ where $h \in \Gamma(S^2 \T)$.
\item $[\lambda , h] = \lambda h$. 
\end{itemize}
\label{def:Lconf}
\end{dfn}

We comment on the cohomology of the global sections of the local dg Lie algebra.
The cohomology in degree zero consists of pairs $(X,\lambda)$ of a vector field and a smooth function, together satisfying the equation $L_Xg + \lambda g = 0$. This is precisely the conformal Killing vector field equation, so that the cohomology in degree zero consists precisely of the Lie algebra of conformal Killing vector fields on the manifold $(M,g)$.

In degree one, the cohomology consists of symmetric two-index tensor fields, which are perturbations of the (inverse) metric. 
These are considered modulo perturbations of the form $L_X g$---so up to those perturbations induced by diffeomorphisms of~$M$---and furthermore up to perturbations of the form $\lambda g$, arising from Weyl rescalings of the background metric. The cohomology in degree one thus precisely corresponds to the space of deformations of the conformal class $[g]$. As usual in deformation theory, symmetries (in degree zero) and deformations (in degree one) fit together into the same derived moduli problem. 

A conformal field theory is a field theory equipped with an action of $\cL_\Weyl$; the Noether currents associated to this action describe the stress tensor at the level of factorization algebras~\cite[Part III]{CG2}. Classes in~$H^1_\loc(\cL_\Weyl)$ correspond to conformal anomalies~\cite{BCRRanomaly}.

            \numpar[p:Riem][Frames] 
            We work in the context of $n$-dimensional smooth manifolds $M$, which we always assume to be parallelizable for simplicity. Fix, once and for all, an $n$-dimensional vector space $V$.           Recall that a \emph{frame} on $M$ is an isomorphism
           \deq{
               e: \ul{V} \to TM
           }
           of vector bundles. Here $\ul{V}$ denotes the trivial bundle $V \times M$. Equivalently, $e$ is a section of $T \otimes V^\vee$ that is everywhere of full rank. (If $M$ is not parallelizable, global frames do not exist, and one must work locally.)

           Given a frame $e$, the dual frame $e^\vee$ gives a map from $T^* M$ to $\ul{V}^\vee$. Since the frame is invertible, we also have maps 
           \deq{
               e^{-1} : TM \to \ul{V}, \quad (e^{-1})^\vee: \ul{V}^\vee \to T^* M.
           }
           (We will avoid making explicit use of the inverse frame in what follows.)

\numpar[p:gen-frame][Reduction of the structure group] 
Before studying a frame model of conformal structures, we begin by examining a slightly more general question.
We define a local dg Lie algebra that encodes the moduli problem of frames compatible with a reduction of the structure group. Fix the datum of an arbitrary Lie algebra $\fm$ equipped with a map $\varphi$  to~$\lie{gl}(V)$.
 \begin{dfn}
                The \emph{local dg Lie algebra of $\fm$-frames} $\cL_\fm$, defined on any $n$-manifold equipped with a framing $e$ by~$V$, consists of the following graded vector bundle:
                \begin{itemize}
                    \item in degree zero, the direct sum of the tangent bundle $T$ and the trivial bundle $\ul{\fm}$, viewed as a subbundle of the trivial bundle $\End(\ul{V})$;                
                    \item in degree one, the tensor product bundle $T \otimes V^\vee $.
                \end{itemize}
                We use the notation $(X,m)$ for a section of the bundle in degree zero, and $f$ for a section of the bundle in degree one. The Lie bracket on $\cL_\fm$ is defined as follows:
                \begin{itemize}
                    \item Vector fields $X$ act by Lie derivative everywhere, using the trivialization of~$\ul{V}$. (We thus regard $\rho$ as a smooth function valued in~$\lie{so}(V)$ and $f$ as a vector field valued in~$V^\vee$.)
                    \item Local $\fm$ transformations act locally on $f$ as they do on $V^\vee$, via the map $\varphi$, and have the commutators amongst themselves appropriate for smooth functions valued in~$\fm$.
                   \end{itemize}
                The differential on~$\cL_\fm$ is the adjoint action of the degree-one element $e$.
            \end{dfn}
            We interpret $\cL_\fm$ as a description of the formal moduli problem of {perturbations} of the fixed background frame $e$, considered up to transformations arising from the reduction of the structure group from $\lie{gl}(d)$ to~$\fm$. Such a perturbation is a general section $f$ of the bundle $TM \otimes V$, which we think of as determining the (linear) formal family $e + tf$ of frames over the formal disk $\Spec \C[\![t]\!]$.

We briefly remark on a couple of degenerate examples that are instructive for the intuition. When $\fm = 0$, we recover a local dg Lie algebra describing framed manifolds. The degree-zero cohomology consists of those vector fields that preserve the background frame; the degree-one cohomology consists of arbitrary perturbations of the frame, modulo those arising from diffeomorphisms. A frame is thought of as a ``reduction of the structure group to zero'' or an absolute parallelism.

When $\fm = \lie{gl}(d)$ and the map $\varphi$ is the identity, we recover a local dg Lie algebra that is clearly quasi-isomorphic just to smooth vector fields in degree zero. We think of this as a description of the moduli problem of manifolds equipped with a ``reduction of the structure group to $GL(d)$,'' or equivalently with no geometric structure at all: $H^1$ is trivial, so there are no moduli, and all infinitesimal diffeomorphisms appear as symmetries in degree zero.

Lastly, when  $\fm = \lie{gl}(d) \oplus \fh$ for an arbitrary finite-dimensional Lie algebra $\fh$ and the map $\varphi$ is projection on the first factor, we obtain a semidirect product Lie algebra in degree zero, consisting of $\fh$-valued functions as a normal subalgebra and viewed as a module over smooth vector fields acting by Lie derivative. We think of this as modelling the formal moduli problem of $d$-manifolds equipped with principal $H$-bundles, where $\Lie(H) = \fh$.

\numpar[p:frmet][Conformal structures via frames]
Recall that a conformal structure is equivalent to a reduction of the structure group from $GL(n)$ to $O(n) \times \R_+$. 
We thus anticipate that $\cL_\fm$ will describe a model for conformal structures when $\fm = \lie{so}(d) \oplus \fz$, where the second summand denotes the center $\fz \cong \lie{gl}(1) $ of $\lie{gl}(V)$. 

To better understand the connection to metrics,
we equip the dual space of~$V$ explicitly with an inner product $\eta$. In other words, we choose a fixed element $\eta \in S^2(V)$ of maximal rank; this reflects our preference to work with inverse metrics in this setting. This datum makes $V$ into a local model for a Riemannian manifold, and identifies an appropriate subalgebra $\fm(\eta)$ of $\lie{gl}(V)$ as identified above.  We will write $\cL_\conf$ for the corresponding local dg Lie algebra $\cL_{\fm(\eta)}$ of $\fm(\eta)$-frames.

           Given $\eta$ together with a $V$-frame on $M$, we can equip $M$ with an (inverse) Riemannian metric. We regard this datum as equivalent to an invertible, self-dual map
           \deq{
               \sharp : T^*M \to TM. 
           }
           By abuse of notation, $\eta$ is an invertible, self-dual map from $V^\vee$ to $V$. A frame then determines a metric via the rule
           \deq[eq:esharp]{
               \sharp = e \circ \eta \circ e^\vee.
           }
           The resulting bundle map is clearly invertible and self-dual. We will not indicate composition explicitly in the following.

            \numpar[p:comparison][Comparing the frame model to $\cL_\Weyl$]
The formal family of frames in~\S\ref{p:gen-frame} determines a (quadratic) formal family of metrics according to the rule~\eqref{eq:esharp}; the family takes the form
           \deq[eq:sharpT]{
               \sharp_t = \sharp(e+tf) = e    \eta    e^\vee + t \left( f    \eta    e^\vee + e    \eta    f^\vee \right) + t^2 \left( f    \eta    f^\vee \right).
           }
This gives us a candidate map from $\cL_\conf$ to $\cL_\Weyl$, which we will verify is an equivalence.

            \begin{prop}
          The cochain map
                \deq{
                \phi^{(1)}: 
                X \mapsto X,  \quad 
                \lambda \mapsto 2\lambda, \quad 
                 \rho \mapsto 0 , \quad 
                 f \mapsto f \eta e^\vee + e \eta f^\vee,
                } 
                and the quadratic correction 
                \deq{
                \phi^{(2)}: (\cL_\conf)^{\otimes 2} \to \cL_\Weyl[1], 
                \quad
                (f_1,f_2) \mapsto \frac{1}{2}
                \left(f_1 \eta f_2^\vee + f_2 \eta f_1^\vee \right).
                }
                together define an $L_\infty$ quasi-isomorphism
                \deq{
                    \phi: \cL_\conf \leadsto \cL_\Weyl,
                }
                witnessing the equivalence of the moduli problems $\cL_\conf$ and~$\cL_\Weyl$.
            \end{prop}

\begin{proof}
The nontrivial $L_\infty$-relation that needs to be checked is the following
\beqn
\phi^{(1)}([X,f]) \overset{?}{=} [X, \phi^{(1)}(f)] + \phi^{(2)} ([e,X], f) .
\eeqn
The left hand side is
\beqn
[X,f] \eta e^\vee + e \eta [X,f]^\vee .
\eeqn
The right hand side is
\beqn
[X, f \eta e^\vee + e \eta f^\vee] + [e,X] \eta f^\vee + f \eta [e,X]^\vee .
\eeqn
Equality follows from the Leibniz rule.

Now, note that the degree one part of $\cL_{\conf}$ is isomorphic to the trivial bundle $V \otimes V^* \cong V \otimes V$ and the degree one part of $\cL_{\Weyl}$ is isomorphic to the trivial bundle $S^2 V$.
This implies that the cone of the cochain map $\phi^{(1)}$ of vector bundles is the complex of vector bundles
\beqn
\wedge^2 V = \lie{so}(V) \to V \otimes V \to S^2 V .
\eeqn
This complex is exact, and the result follows.
\end{proof}

            \subsection{Superconformal structures and reductions of the structure group}
            
            \numpar[p:lem-gradings][An ideal in $\Conf(\fn)_+$]
            Recall our grading conventions for~$\Conf(\fn)$ as outlined in~\S\ref{grDer} and~\S\ref{sec:grading}; the weights of generators are summarized in Table~\ref{tab:moredegs}. The weight is bounded from below by zero, and determines the homological degree on the super dg Lie algebra $\Conf(\fn)$; the $\theta$-degree is bounded from above by $+2$, and determines the total parity modulo two. The intrinsic parity (the $\Z/2\Z$ grading on the super dg Lie algebra $\Conf(\fn)$) is  thus determined by the totalized degree---the weight minus the $\theta$-degree---modulo two. The totalized degree is bounded from below by $-2$.
            \begin{lem}
            \label{lem:ideal}
            Consider the sub dg Lie algebra $\Conf(\fn)_+ \subset \Conf(\fn)$ consisting of elements of even parity. The sub dg Lie algebra $I$ spanned by all elements of $\Conf(\fn)_+$ with strictly positive totalized degree is an ideal in $\Conf(\fn)_+$.
            \begin{proof}
            Since totalized degree determines intrinsic parity, $\Conf(\fn)_+$ is spanned by all summands of even totalized degree, and $I$ is spanned by all summands with totalized degree $\geq 2$. The Lie bracket has terms of totalized degree $0$ and $+2$ that are zeroth-order and first-order differential operators, respectively.
            It is obvious that bracketing with terms of totalized degree zero preserves $I$. It remains only to check that this is true for terms of totalized degree $-2$. But since these summands are just vector fields on spacetime, they can only participate in the Lie bracket via first-order differential operators, thus via terms of totalized degree $+2$.
             Since the differential on~$\Conf(\fn)$ is given by the adjoint action of~$\d_0$, which is an element of totalized degree zero, $I$ is a differential ideal.
            \end{proof}
            \end{lem}
            We remark that the lemma is \emph{not} true for the full super dg Lie algebra $\Conf(\fn)$, which in general contains no obvious ideal. It fails due to summands in totalized degree $-1$ (corresponding to local supersymmetries), which can participate in the Lie bracket via terms of totalized degree zero.

            \numpar[p:boson-comp][$\fg_0$-frames from $\Conf(\fn)$]
            We construct a map from the bosonic part of the formal super moduli problem of superconformal structures to the formal moduli problem of $G_0$-structures.
\begin{thm}
Let $\fn$ be a supertranslation algebra for which the bracket map $\gamma$ is surjective, and $\fg_0$ the Lie algebra of degree-zero derivations of~$\fn$. 
 There is a strict map of dg Lie algebras from $\Conf(\fn)_+$ to $\Conf(\fn)_+/I$. The dg Lie algebra $\cL_{\fg_0}$ of $\fg_0$-frames is a quasi-isomorphic sub dg Lie algebra of $\Conf(\fn)_+/I$.   
  \label{thm:s/conf}
    \begin{proof}
The first statement is an obvious consequence of Lemma~\ref{lem:ideal}. To prove the rest, we construct a strict map of dg Lie algebras
\deq{
 \phi: \cL_{\fg_0} \to \Conf(\fn)_+/I .
 }
To construct this map, observe that there exists a copy of $\lie{gl}(\Sigma)$ in bidegree $(0,0)$ in $\Conf(\fn)_+$, spanned by the linear vector fields
\deq{
\lambda \pdv{ }{\lambda} + \theta \pdv{ }{\theta}.
}
Using the map $\rho_1$, we can map $\fg_0$ to the corresponding subalgebra of this copy of $\lie{gl}(\Sigma) \subset \Conf(\fn)_+$.
 It is also clear that we can map the smooth vector fields in $\cL_{\fg_0}$ to $\Conf(\fn)_+$ along the representatives $\partial/\partial x$ in bidegree $(0,2)$, 
Lastly, using the map $\gamma^\vee$, we also find a copy of $V^\vee$ sitting $\fg_0$-equivariantly inside of $\Sigma^\vee \otimes \Sigma^\vee$, so that we can map sections of ${V}^\vee \otimes T$ to $\Conf(\fn)_+$ along the representatives $(\lambda \gamma \theta) \cdot \partial/\partial x$ in bidegree $(1,1)$.
 
It is straightforward to check that bracketing with vector fields reproduces the appropriate Lie bracket in~$\cL_{\fg_0}$, and that the brackets of the representatives of local $\fg_0$ transformations are as they should be. The bracket of these representatives with the representatives of a frame perturbation $f \in \Gamma(V^\vee \otimes T)$ has two terms, one consisting of the obvious action of $\fg_0$-valued functions on the frame in bidegree $(1,1)$, the other containing the derivative of the $\fg_0$-valued function along the frame in bidegree $(1,-1)$. But the latter is contained in $I$, so that the bracket in $\Conf(\fn)_+/I$ again reproduces that in $\cL_{\fg_0}$. Similarly, the bracket of two frame representatives only contains terms that are of order one in spacetime derivatives; these land in bidegree $(2,0)$, and so are also contained in~$I$.

Since the differential is just the adjoint action of~$\d_0$, it follows that $\phi$ is a strict map of dg Lie algebras.  It remains to check that this map is an isomorphism, but this follows straightforwardly from the computation of the cohomology of $\Conf(\fn)$ in these degrees as presented in Theorem~\ref{thm: univ}.
    \end{proof}
\end{thm}
In words, there is a map from the bosonic part of the formal super moduli problem of superconformal structures based on $\fn$ to the formal moduli problem of manifolds equipped with a reduction of structure group to the automorphisms of~$\fn$.

\numpar[p:cor][Conformal structures from frames]
It is clear that, in the case where $\rho_2(\lie{g}_0) = \lie{so}(d) \oplus \fz$, we have a map of formal moduli problems from $\cL_{\fg_0}$ to~$\cL_\conf$. Putting this together, we have the following:
\begin{cor}
Let $\fn$ be a supertranslation algebra for which the bracket map $\gamma$ is surjective and 
$\rho_2(\fg_0) = \lie{so}(d) \oplus \fz$. 
Then there is a strict map of dg Lie algebras from  $\Conf(\fn)_+$ to $\cL_\conf$.
\label{cor:s/conf}
\end{cor}

       \section{Deformations of superspace are conformal supergravity}
\label{sec: examples}

In this section, we summarize results on the (component-field) multiplets $\mu\Conf(\lie{n})$.
We will present these tabularly in some physically meaningful examples. 
This means that we will take the super Lie algebra $\lie{n}$ to be a standard supersymmetry algebra.

For the sake of space, we are unfortunately compelled to depart from the conventions of~\cite{perspectives}. 
For us, the vertical axis of the table corresponds to the weight grading, and the horizontal to the homological degree as defined in~\S\ref{sec:grading}. (In~\cite{perspectives}, the horizontal axis corresponded to the \emph{totalized} degree; thus our tables are obtained from those conventions by shifting row $i$ to the left by $i$ steps.) Degrees increase from left to right and from top to bottom. Note that, in spite of the names, the weight grading determines the \emph{cohomological} grading (or ghost number) of the fields in the multiplet. All results are ambiguous up to an overall shift, which is determined by physical considerations and with which we will not  concern ourselves explicitly.

In our tables, we will just display the underlying vector bundle 
of the multiplet.
We indicate the differential only schematically, or not at all; the arrows we draw indicate terms that are present, but are not necessarily an exhaustive description of the transferred $D_\infty$ structure on $\mu \Conf(\fn)$.
Arrows do not necessarily represent bundle maps, but rather maps built from differential operators acting on sections.
Bundles are always complexified, and $\ul{\C}$ will denote the trivial vector bundle with fiber $\C$.
Thus, for example, the two-term complex $\Omega^0 \xto{\d} \Omega^1$ will be represented, in our tables, by $\ul{\C} \to T^*$, with the de Rham differential understood.

\subsection{One- and two-dimensional superconformal algebras}
We begin by analyzing theories in one and two dimensions. Thanks to holomorphic factorization in two dimensions, the theories are essentially equivalent; the varieties in question are all quadric hypersurfaces, so that the algebraic geometry is comparatively trivial. 
\numpar[sqm][Supersymmetric quantum mechanics]
        In one spacetime dimension with $\cN$ supercharges, the nilpotence variety is defined by the single quadratic equation
        \begin{equation}
        	\lambda_1^2 + \dots + \lambda_\cN^2 = 0 .
        \end{equation}
        Up to a factor of two, the map $\lambda \gamma$ is simply given by the matrix
        \begin{equation}
        	(\lambda_1, \dots, \lambda_\cN),
        \end{equation}
        whose cokernel is just $\CC$. 
        It is well-known that the augmentation ideal of the nilpotence variety (the structure sheaf at the cone point) corresponds to the free superfield~\cite{perspectives}.
        Hence $\Conf(\fn)$ can be identified---as a multiplet---with the free superfield. 

The local Lie algebra structure is somewhat more interesting, and is closely related to a familiar superconformal algebra.
Let $K(1|\cN)$ be the Lie algebra of contact vector fields on $\R^{1|\cN}.$\footnote{This is the positive part of the non-centrally extended superconformal algebra that Kac calls $K(1|N)$ in \cite{KacSusy}.}
We can think of this as a local Lie algebra on $\R$.

        \begin{thm}
        \label{thm:sqm}
            Let $\fn$ be the supertranslation algebra for $\N$-extended supersymmetric quantum mechanics.
            Then there is an isomorphism 
            \deq{
                \mu \Conf(\fn) \to K(1|\N).
            }
            \begin{proof}
                This is an easy observation. $K(1|\N)$ is defined to be those vector fields that preserve the contact one-form
                \[
                    v   = \d t + \sum_{i=1}^\N \theta_i\, \d \theta_i
                \]
                up to scale. In this example, this one-form is the unique left-invariant even one-form---compare~\eqref{eq:linv1}. By the standard theory of contact structures, vector fields preserving the one-form  $v$ up to scale are identical with vector fields preserving the distribution $\ker(v)$, which is identified with the $(0|\N)$-dimensional distribution $D$. Since $\mu \Der(A^\bu)$ is supported in cohomological degree zero, strict distribution-compatible vector fields agree exactly with our derived model for distribution-compatible vector fields here. The identification of $K(1|\N)$ with an unusual super Lie algebra structure on the space of functions on~$\C^{1|\N}$ was pointed out by Cheng and Kac~\cite[\S1.2]{ChengKac}.
            \end{proof}
        \end{thm}

        \numpar[twod][Two-dimensional theories]
        In two dimensions, the supersymmetry algebra splits as a direct sum of a holomorphic and an anti-holomorphic algebra, each of which is a generalized supertranslation algebra identical to that of supersymmetric quantum mechanics. If we work with chiral two-dimensional theories, the story is thus identical to that in the previous subsection. 
        Thus, for chiral theories with supersymmetry of type $(\cN,0)$ the multiplet $\mu \Conf(\fn)$ is equivalent to $K(1|\N)$, now thought of as a local Lie algebra on $\C$.
        More precisely, the $\infty$-dimensional locally compact super Lie algebra that is referred to as $K(1|\N)$ is the formal completion of this local Lie algebra at the origin.
        This difference is whether one considers smooth/holomorphic sections or formal power series.
        For example, in our conventions, $K(1|1)$ is the local Lie algebra whose even part is the Dolbeault complex of the vector bundle $\T^{1,0}_\C$ and whose odd part is the Dolbeault complex of the vector bundle $K^{-1/2}_\C$.
        For theories with non-chiral supersymmetry of type $(\cN,\cN)$, there is also a complex conjugate of the multiplet $\mu \Conf(\fn)$ present.
        
        When $\cN=k\leq 3$, $\mu \Conf(\fn)$ is equivalent to the $\cN=k$ superconformal multiplet.
        When $\cN=4$, $\mu \Conf(\fn)$ contains the ``big'' $\cN=4$ superconformal multiplet.
        More precisely, as a super Lie algebra it contains the ``big'' $\cN=4$ algebra as a codimension one subalgebra.
        
        \subsection{Three dimensions}
        \label{ssec: 3d}
        Recall that, in Euclidean signature, $\Spin(3) = SU(2)$. The unique irreducible spin representation $S$ is the defining representation of~$SU(2)$, and the three-dimensional vector representation $V$ is the adjoint representation. There is a single irreducible spin-$3/2$ representation $\Psi$, defined by the decomposition
        \[
V \otimes S = S \oplus \Psi
        \] and having Dynkin label $[3]$. The $\N$-extended supertranslation algebra is of the form
        \begin{equation}
        	S(-1) \otimes U \oplus V(-2),
        \end{equation}
        where $U \cong \CC^\cN$ is equipped with a non-degenerate symmetric bilinear form $g$. The bracket is induced by the isomorphism $\Sym^2(S) \cong V$ tensored with~$g$.

          \numpar[3dN=1][Minimal supersymmetry]
        The defining ideal of the nilpotence variety is
        \begin{equation}
        	I = (\lambda_1^2 , \lambda_1 \lambda_2, \lambda_2^2) \subset \Sym^\bu(S);
        \end{equation}
        the corresponding multiplet was described in~\cite{perspectives}. The map $\varphi$ is explicitly given by the matrix
        \begin{equation}
        	\begin{bmatrix}
        	2 \lambda_1 & 0 \\
        	\lambda_2 & \lambda_1 \\
        	0 & 2\lambda_1
        	\end{bmatrix} .
        \end{equation}
        The field content of $\mu\Conf(\fn)$ can be summarized with the table 
        \begin{equation}
            \begin{tikzcd}[column sep = 2 ex, row sep = 2 ex]
                T \ar[d] & S \ar[d] \\
                \Sym^2_0(T^*) & \Psi.
        \end{tikzcd} 
        \end{equation}
        We note that, in this instance, the module $\coker(\varphi)$ coincides with the conormal module $I/I^2$~\cite{perspectives}.
The first arrow is Lie derivative of the fixed metric followed by projection onto the traceless part.
The second arrow is the Penrose operator.

        \numpar[3dN=2][The case $\cN=2$]
        For $\cN=2$ one finds the following field content for $\mu \Conf(\fn)$:
        \begin{equation}
            \begin{tikzcd}[row sep = 2 ex, column sep = 2 ex]
                T \ar[d] & S \otimes U \ar[d] & \ul{\C} \ar[d] \\
                \Sym^2_0(T^*) & \Psi \otimes U  & T^* .
        \end{tikzcd} 
        \end{equation}
An $R$-symmetry connection appears in the multiplet for the first time here. 
From the physical (off-shell) degrees of freedom
\beqn
(h_{\mu\nu}, \psi_\nu^i, A_\mu)
\eeqn
read off from the second line, modulo gauge equivalence determined by the top line,
we find agreement with the three-dimensional $\cN=2$ Weyl multiplet~\cite{BHRT}.

        \numpar[3dN=4][The case $\cN=4$]
        In this case, the field content of $\mu \Conf(\fn)$ looks as follows (with $U = \C^4$ and $\Sigma = S \otimes U$):
        \begin{equation}
            \begin{tikzcd}[row sep = 2 ex, column sep = 2 ex]
                T \ar[d] & \Sigma \ar[d] & \lie{so}(4) \ar[d] \\
                \Sym^2_0(T^*) & \Psi \otimes U  & (T^* \otimes \lie{so}(4) ) \oplus \ul{\C} & \Sigma & \ul{\C}.
        \end{tikzcd} 
        \end{equation}
The multiplet now contains additional matter fields (fermions $\chi^i \in C^\infty \otimes \Sigma$, and a bosonic scalar $D \in C^\infty$). 
Nevertheless, it is still supported in degrees zero and one (ghost number $-1$ and $0$ in physics conventions).
Again, we find agreement with the three-dimensional $\cN=4$ Weyl multiplet \cite{BHRT}.

        \numpar[3dN=8][The case $\cN=8$]
        Our computation reproduces---in fact, is identical to---the construction of Cederwall, Gran, and Nilsson~\cite{Ced-dragon}. (Those authors analyze the pure spinor superfield associated to the sheaf of surviving supertranslations, working in that specific example.) 
\begin{equation}
    \begin{tikzcd}[row sep = 2 ex, column sep = 0.4 ex]
                T \ar[d] & \Sigma \ar[d] & \lie{so}(8) \ar[d] \\
                 \Sym^2_0(T^*) & \Psi \otimes U  & 
                 \Vectorstack{{T^* \otimes \lie{so}(8)} {\wedge^4 U}} & (S \otimes \wedge^3 U)^{\oplus 2} & 
                 \Vectorstack{{\wedge^4 U} {\wedge^2 T^* \otimes \lie{so}(8)}} \ar[d] & \Psi \otimes U \ar[d] & \Sym^2_0(T^*) \ar[d] \\
                & & & & T  & \Sigma  & \lie{so}(8)  \\
        \end{tikzcd} 
        \end{equation}
        We emphasize that the multiplet does \emph{not} admit the structure of a BV theory---not even a $\Z/2\Z$-graded BV theory. It is isomorphic to a zero-shifted cotangent bundle. 
        It is the base of this cotangent bundle which agrees with the three-dimensional $\cN=8$ multiplet described in \cite{BHRT}.
        Furthermore, it is no longer supported in degrees zero and one: ``physical'' fields (in cohomological degree one in our conventions) 
        are subject to further ``differential constraints,'' imposed by fields in degree two.

\subsection{Four dimensions}

        We identify $\Spin(4) \cong SU(2) \times SU(2)$; as throughout, we denote the two-dimensional chiral spin representations by $S_\pm$, and the four-dimensional Dirac spin representation by $S = S_+ \oplus S_-$. There are two chiral ``spin-$3/2$'' representations $\Psi_\pm$ of dimension six, defined by the equations
        \[
            V \otimes S_\pm \cong \Psi_\pm \oplus S_\mp.
        \]
        The Dynkin labels are $\Psi_+ = [21]$ and $\Psi_- = [12]$. By analogy with the Dirac spinor, we write $\Psi = \Psi_+ \oplus \Psi_-$.

        The supertranslation algebra is of the form
        \begin{equation}
            \left(S_+ \otimes U \oplus S_- \otimes U^\vee \right)(-1) \oplus V(-2) ,
        \end{equation}
        where the bracket is induced by the isomorphism $V \cong S_+ \otimes S_-$ and the natural pairing between $U = \C^\N$ and its dual.
        
        \numpar[4dN=1][Minimal supersymmetry]
        One finds the following field content, which matches the $\cN=1$ conformal supergravity multiplet as described in~\cite{SugraBook}:
        \begin{equation}
            \begin{tikzcd}[column sep = 2 ex, row sep = 2 ex]
                T \ar[d] & S  \ar[d] & C^\infty \ar[d] \\
        \Sym^2_0(T^*) & \Psi & T^*
        \end{tikzcd}
        \end{equation}
        
\numpar[4dN=2][The case $\N=2$]
        The component fields match the ``Weyl multiplet'' of conformal supergravity, as discussed in~\cite{4dConfSugra}. They can be summarized in the following table:
        \begin{equation}
            \begin{tikzcd}[column sep = 2 ex, row sep = 2 ex]
               T \ar[d] & \Vectorstack{{S_+ \otimes U} {S_- \otimes U^\vee}} \ar[d] & \lie{gl}(2)_R \ar[d] \\
               \Sym^2_0(T^*) & \Vectorstack{{\Psi_+ \otimes U} {\Psi_- \otimes U^\vee}} & \Vectorstack{ {T^* \otimes \lie{gl}(2)_R} {\wedge^2 T^*} } & \Vectorstack{ {S_+ \otimes U} {S_-\otimes U^\vee} } & \ul{\C}
        \end{tikzcd}
        \end{equation}
We see the auxiliary fields of the Weyl multiplet appearing at higher $\theta$-degree. 
        The multiplet still contains no fields of cohomological degree higher than one, and thus no constraints. 
        \numpar[4dN=4][The case $\N=4$]
        The component fields are summarized in Table~\ref{tab:4d}, where now $U = \C^4$. We use Dynkin labels for the algebra $A_3$, with the convention that $U = [001]$. 
      \begin{table}
        \begin{equation*}
            \begin{tikzcd}[column sep = 0.8 ex, row sep = 1.2 ex]
               T \ar[d] & \Vectorstack{{S_+ \otimes U} {S_- \otimes U^\vee}} \ar[d] & \lie{gl}(4)_R \ar[d] \\
                \Sym^2_0(T^*) & \Vectorstack{{\Psi_+ \otimes U} {\Psi_- \otimes U^\vee} {S_+ \otimes U^\vee} {S_- \otimes U}} & \Vectorstack{{T^* \otimes \lie{gl}(4)_R} {(T^*)^{\oplus2}} {\wedge^2 T^* \otimes \wedge^2 U} {(\wedge^2 U)^{\oplus2}} {\Sym^2 U \oplus \Sym^2 U^\vee}} & \Vectorstack{{(S_+ \otimes U)^{\oplus 2}} {(S_-\otimes {U^\vee})^{\oplus2}} {S_+ \otimes [110]} {S_- \otimes [011]}} & {[020]} \\
       & & \wedge^2 T^* & \Vectorstack{{(S_+ \otimes U^\vee)^{\oplus 2}} {(S_- \otimes U)^{\oplus 2}}} & \Vectorstack{{\wedge^3 T^*} {(\wedge^2 U)^{\oplus 2}}} \\
      & & & & \wedge^4 T^*.
        \end{tikzcd}
        \end{equation*}
        \caption{$\mu \Conf(\fn)$ for four-dimensional $\N=4$ supersymmetry}
        \label{tab:4d}
        \end{table}

        \subsection{Six dimensions}
        \ytableausetup{smalltableaux}
        In six dimensions, there is an exceptional isomorphism identifying $\Spin(6) \cong SU(4)$; under this identification, the two spinor representations $S_+$ and $S_-$ correspond to the fundamental and antifundamental representations. There are $\Spin(6)$-equivariant isomorphisms $\wedge^2 S_\pm \cong V$, where $V$ denotes the six-dimensional vector representation. There are also irreducible chiral spin-$3/2$ representations $\Psi_\pm$, defined by the decomposition
        \[
            V \otimes S_\pm = S_\mp \oplus \Psi_\pm.
        \]
        They have Dynkin labels $[101]$ and $[110]$, dimension 20, and can also be thought of as the component in the tensor cube of $S_\pm$ with ``hook''-type symmetry $\ydiagram{2,1}$.

The six-dimensional supertranslation algebra of type $(\cN,0)$ is
        \begin{equation}
        \fn = (S_+ \otimes U)(-1) \oplus V(-2) \: ,
        \end{equation}
        where $U = \C^{2\N}$ is a symplectic vector space with antisymmetric bilinear form $\omega$. The bracket is $\wedge \otimes \omega$. The R-symmetry group is $\Sp(\N)$.
        
\numpar[6d1][$\N=(1,0)$ supersymmetry]
For the first time among cases with minimal supersymmetry, the cohomology contains fields other than the frame, gravitino, and $R$-connection. It takes the form
\begin{equation}
       \begin{tikzcd}[column sep = 2 ex, row sep = 2 ex]
        T \ar[d] & S_+ \otimes U \ar[d] & \lie{sp}(1) \ar[d]  \\
        (\Sym^2 T^*)_0 & \Psi_+ \otimes U & \Vectorstack{{T^* \otimes \lie{sp}(1)} {(\wedge^3 T^*) _-}} &  S_+ \otimes U & \ul{\C}
    \end{tikzcd}
    \label{eq:sixmin}
\end{equation}
where $(\wedge^3 T^*)_-$ denotes the anti-self-dual part of $\wedge^3 T^*$.
This precisely matches the Weyl multiplet in six-dimensional minimal supersymmetry, as found in~\cite{BSvP6d}. (See also~\cite{Conf6Sup} for a more recent superspace approach.)

\numpar[6d2][$\N=(2,0)$ supersymmetry]

Let $\lie{n}$ be the six-dimensional $\cN=(2,0)$ supertranslation algebra.
The component fields of $\Conf(\lie{n})$ precisely recover those of the $\N=(2,0)$ conformal supergravity multiplet~\cite{BSvP}. 
In our conventions, the minimal presentation $\mu \Conf(\fn)$ of the local dg Lie algebra takes the following form:
\begin{equation}
       \begin{tikzcd}[column sep = 2 ex, row sep = 2 ex]
        T \ar[d] & S_+ \otimes [01] \ar[d] & \lie{sp}(2) \ar[d]  \\
        (\Sym^2 T^*)_0 & \Psi_+ \otimes [01] & \Vectorstack{{T^* \otimes \lie{sp}(2)} {(\wedge^3 T^*)_- \otimes [10]}} &  S_+ \otimes [11]  & {[20]}
    \end{tikzcd}
    \label{eq:A1untwist}
\end{equation}
Here, we use $B_2$ Dynkin labels, so that $[10]$ is the five-dimensional vector representation of $\Spin(5)$ and $[01]$ the four-dimensional spin representation of $\Spin(5)$---equivalently, the defining representation of $\lie{sp}(2)$. The adjoint representation is $\lie{sp}(2) = [02]$.

\subsection{Ten dimensions}

We focus only on minimal chiral supersymmetry in this dimension. 
The supersymmetry algebra is
\beqn
\lie{n} = S_+ (-1) \oplus V (-2)
\eeqn
where $V$ is the ten-dimensional vector representation and $S_+$ is the $16$-dimensional chiral spin representation of $\op{Spin}(10)$.
%
The linearized description of the multiplet $\mu \Conf(\lie{n})$ is presented in Table~\ref{tab:10d}.
\begin{table}
\begin{equation*}
       \begin{tikzcd}[column sep = 2 ex, row sep = 2 ex]
        T \ar[d] & S_+ \ar[d]  \\
        (\Sym^2 T^*)_0 & \Vectorstack{{\Psi_+} {S_-}} & \Vectorstack{{\wedge^6 T^*} {T^*}} \\
        & & \Vectorstack{{\ul{\C}} {\wedge^2 T}} & S_-.
    \end{tikzcd}
    \label{eq:tab10d}
\end{equation*}
\caption{$\mu \Conf(\fn)$ for ten-dimensional minimal supersymmetry}
\label{tab:10d}
\end{table}
This table is in agreement with the ten-dimensional conformal supergravity multiplet described in \cite{ConfSugra10d,MFOF10d}.
There are bosonic degrees of freedom given by the one-form $G \in \Omega^1$ and a six-form $C \in \Omega^6$.
Differential constraints appear in degree two, though we do not work out their form in detail here.

\subsection{Eleven dimensions}

Finally we address eleven-dimensional supersymmetry.
The supersymmetry algebra is
\beqn
\lie{n} = S(-1) \oplus V(-2)
\eeqn
wher $S$ is the unique spin representation and $V$ is the vector representation.

\numpar[p: 11d][Minimal supersymmetry]
We recover the cohomology of~\cite[Table 3]{CederwallM5}. (Indeed, this identification was already made in~\cite[\S3.1 and Table 2]{Ced-towards}.)
We present the cohomology in~Table~\ref{tab:11d}.
\begin{table}
\begin{equation*}
       \begin{tikzcd}[column sep = 2 ex, row sep = 2 ex]
        T \ar[d] & S \ar[d]  \\
        (\Sym^2 T^*)_0 & {{\Psi}\oplus {S}} & \Vectorstack{{\wedge^3 T^* \oplus T}} \\
        &  {{\Psi}\oplus {S}} & \Vectorstack{{(\Sym^2 T^*)_0 \oplus \ul{\C}} {T^* \oplus \wedge^2 T^*} {\wedge^3 T^* \oplus \wedge^5 T^*}} \\
        &  & \Vectorstack{{(\Sym^2 T^*)_0 \oplus \ul{\C}} {T^* \oplus \wedge^2 T^*} {\wedge^3 T^* \oplus \wedge^5 T^*}} & {{\Psi}\oplus {S}} \\
        & & \Vectorstack{{\wedge^3 T^* \oplus T}}  & {{\Psi}\oplus {S}} \ar[d]  & (\Sym^2 T^*)_0  \ar[d] \\
         & & & S & T.
    \end{tikzcd}
    \label{eq:tab11d}
\end{equation*}
\caption{$\mu \Conf(\fn)$ in eleven dimensions}
\label{tab:11d}
\end{table}

We give a few remarks on the interpretation in this case. In our analogy between almost-complex manifolds and superspaces, eleven-dimensional superspace is a Calabi--Yau twofold, and the interactions of full eleven-dimensional supergravity are described in the BV formalism by the holomorphic Poisson bracket on~$A^\bu = W^{0,\bu}$ induced by a Calabi--Yau structure~\cite{CY2}. 

The additional structure required to define the theory can be thought of as a trivialization of the sheaf $W^{-2,\bu}$ of ``holomorphic top forms.'' This is clearly analogous to the typical data, which should be some section of an appropriate analogue of the Berezinian line bundle. Choosing such a trivialization induces an isomorphism between one-forms and vector fields. Indeed, in this case, one finds that 
\deq{
\Conf(\fn) \cong W^{-1,\bu} \cong  \Omega_{A^\bu/\C},
}
so that the subtleties identified in~\S\ref{p: oneforms} do not arise.

The vector fields preserving the additional structure should be modeled by the divergence-free (or equivalently symplectic) vector fields; in other words, by a complex of the form
\begin{equation}
\begin{tikzcd}
W^{-1,\bu} \ar[r,"\partial"] & W^{-2,\bu}.
\end{tikzcd}
\end{equation}
But $W^{0,\bu}$ is just a one-dimensional central extension of this! In this sense, the constructions of~\cite{CY2} are a particularly clean instance of the intuition mentioned above in~\S\ref{p:confstr}.
\section{Applications to twisted theories}
\label{sec: twisted}

We move on to examples of the multiplet $\mu \Conf(\lie{n})$ for which $\lie{n}$ is not a super Lie algebra underlying standard supersymmetry.
In this section, we consider \textit{twists} of some of the supermultiplets (and others) discussed in the previous section.
Following~\S\ref{twists}, twists of such supermultiplets arise from taking $\lie{n}$ to be the twist of the original supersymmetry algebra. A full classification of possible twisting supercharges was given in~\cite{NV,ElliottSafronov}.  

Our primary focus is on two types of twists: minimal (holomorphic) twists and \emph{maximal twists} (see~\S\ref{geom} for terminology).
For the latter, all issues related to supergeometry become nullhomotopic in the twist, and we obtain a uniform description in terms of transversely holomorphic vector fields. For the former, we use our results to compute the holomorphic twists of stress-tensor multiplets in superconformal theories, explicitly recovering higher Virasoro algebras~\cite{SCA} and a local version of the exceptional simple super Lie algebra $E(3|6)$.

\subsection{Maximal twists}
As discussed above in~\S\ref{twists}, the results of~\cite{spinortwist} indicate that twisting commutes with the construction of flat superspace. Thus the flat superspace of type~$\fn_Q$ is quasi-isomorphic to the twist by~$Q$ of the flat superspace of type~$\fn$, as indicated in~\eqref{eq:AQAQ}. 
Further studies of twisting at the level of the pure spinor superfield can be found in~\cite{FabianThesis} and~\cite{SJTwist}.

If we start with superspaces of type~$\fn$, having dimension $d|k$ and homological dimension
\deq{
\hdim(N) = h =  d - k + \dim(Y_\fn),
}
we can compute the maximal twist of $\Conf(\fn)$ simply by taking derivations of the cdga $A^\bu(\fn_Q)$. This is permitted because~$A^\bu(\fn_Q)$ is semifree in the maximal twist~\cite{Hinich}. In fact, on flat superspace,
\deq{
A^\bu(\fn_Q) \cong \Omega^{0,\bu}(\C^h) \otimes \Omega^\bu_\dR(\R^d-2h)
}
is a resolution of transversely holomorphic functions on a THF structure, specified by an integrable distribution of dimension~$d-h$ in~$T_\C \R^d$.

It follows that 
the \textit{maximally twisted} superconformal algebra
\deq{
\Conf(\fn)^Q \cong \Vect^\THF(\C^{h} \times \R^{d-2h}) = \Omega^{0,\bu}(\C^h, \T) \otimes  \Omega^\bu_\dR(\R^{d-2h})
}
is equivalent to the local dg Lie algebra resolving the sheaf of transversely holomorphic vector fields.
This is a dg Lie algebra with differential $\dbar + \d_\dR$ and bracket which extends the usual Lie bracket of holomorphic vector fields. We list physical examples to which this general result applies in Table~\ref{tab:maxtwist}. 

\begin{table}
\begin{tabular}{|l|l|l|} \hline
\emph{Superspace} & \emph{Maximal twist} & \emph{$\op{Conf}(\lie{n})^Q$} \\ \hline
3d $\N=2$ & holomorphic twist & $\Vect^\THF(\C \times \R)$ \\ 
\hline
4d $\N=1$ & holomorphic twist & $\Vect^\hol(\C^2)$ \\ 
\hline
4d $\N=2$ & Kapustin twist & $\Vect^\THF(\C \times \R^2)$ \\ 
\hline
4d $\N=4$ & Kapustin--Witten twist & trivial \\ 
\hline
6d $\N=(1,0)$ & holomorphic twist & $\Vect^\hol(\C^3)$ \\ 
\hline
6d $\N=(2,0)$ & nonminimal twist & $\Vect^\THF(\C \times \R^4)$ \\ 
\hline
10d $\N=(1,0)$ & holomorphic twist & $\Vect^\hol(\C^5)$ \\
\hline
10d $\N=(2,0)$ & maximal twist & $\Vect^\THF(\C \times \R^8)$ \\ 
\hline
11 $\N=1$ & nonminimal twist & $\Vect^\THF(\C^2 \times \R^7)$ \\ \hline
\end{tabular}
\caption{Examples of maximal twists}
\label{tab:maxtwist}
\end{table}

\subsection{Holomorphic twists in three dimensions}

Nonzero twisting supercharges for three-dimensional supersymmetry exist when $\N \geq 2$.
In the case $\N=2$ there is a unique such twisting supercharge up to equivalence. 
It is holomorphic in the sense that two directions are invariant. Globally, such twisted theories can be placed on three-manifolds equipped with a rank-two transverse holomorphic foliation~\cite{ACMV}.
Choosing such a holomorphic supercharge $Q$
 is possible for any $\cN\geq 2$, since the corresponding supertranslation algebra always contains an $\cN=2$ subalgebra.

\numpar[p:3dholo][The supertranslation algebra $\fn_Q$]
Recall that the odd elements in the supertranslation algebra take the form 
\deq{
\fn_1 = S \otimes R,
}
 where $R = \C^\cN$ is equipped with a nondegenerate symmetric bilinear form $g$.
The Lie algebra of automorphisms $\fg_0$ is $\lie{sp}(S) \oplus \lie{gl}(1) \oplus \lie{so}(U)$. Without loss of generality, the element $Q$ takes the form $u \otimes v$, where $u$ is a highest-weight vector of $S$ and $v$ a highest-weight vector of~$R$. Then $U \cong \C\cdot v \oplus \C\cdot \bar{v} \oplus \tilde R$, where $\bar{v}$ denotes the opposite weight vector and $\tilde{U} = \C^{\N-2}$  again has a nondegenerate symmetric bilinear form.
A basis for $\fn_2$ is given by $u \otimes u$, $u \otimes \bar{u}$, and~$\bar{u} \otimes\bar{u}$.

The image of~$Q$ under the adjoint action of $\fg_0$ consists of all elements of the form $u' \otimes v$ or $u \otimes v'$, where $u'\in S$, $v' \in \tilde{U}$ are arbitrary elements. Acting on~$\fn_1$, $\ad_Q$ maps $u \otimes \bar{v}$ to $u \otimes u$ and $\bar{u} \otimes \bar{v}$ to $u \otimes \bar{u}$. Thus we have
\deq{
(\fn_Q)_1 = \bar{u} \otimes \tilde{U}, \quad
(\fn_Q)_2 = \C\cdot (\bar{u} \otimes \bar{u}),
}
with the bracket just given by the inner product on~$\tilde{U}$.  

\numpar[p:moduli3d][Deformations of twisted superconformal structures]
Applying Theorem~\ref{thm:sqm}, 
we conclude that
$\Conf(\fn)$ for holomorphically twisted $\N$-extended supersymmetry in three dimensions is quasi-isomorphic to $K(1|\N-2)$, viewed as a local Lie algebra on $\C \times \R$ via the resolution $\Omega^{0,\bu}(\C) \otimes \Omega^\bu(\R)$. 
We note, in particular, that three-dimensional $\N=8$ supersymmetry gives rise to the 
algebra $K(1|6)$ in the holomorphic twist.

\subsection{Holomorphic twists in four dimensions}
\label{ssec:tw4d}

Nonzero twisting supercharges for four-dimensional supersymmetry always exist.
In the case $\N=1$ there is a unique such twisting supercharge of each chirality up to equivalence. The twisting supercharges are holomorphic in the sense that two directions are invariant.
In what follows we choose such a supercharge $Q \in S_+$ (without loss of generality we can take it to be of positive chirality).
Such a supercharge determines a complex structure on $V = \R^4$. We denote the corresponding maximal isotropic subspace by $L \subset V \otimes \C = \C^4$.
Again, this data exists for any extended supersymmetry algebra. 

\numpar[4dnQ][The supertranslation algebra $\fn_Q$]
The odd elements in the supertranslation algebra take the form 
\deq{
\fn_1 = (S_+ \otimes R) \oplus (S_- \otimes R^\vee),
}
 where $R = \C^\cN$ has no additional structure.
The even elements are $\fn_2 = V \cong S_+ \otimes S_-$.
The Lie algebra of automorphisms $\fg_0$ is $\lie{so}(S_+ \otimes S_-) \oplus \lie{gl}(1) \oplus \lie{gl}(U)$; 
we recall the exceptional isomorphism $\lie{so}(S_+ \otimes S_-) \cong \lie{sp}(S_+) \oplus \lie{sp}(S_-)$.

Without loss of generality, the element $Q$ takes the form $u \otimes v$, where $u$ is a highest-weight vector of $S_+$ and $v$ a highest-weight vector of~$R$. Then $R \cong \C\cdot v \oplus \tilde R$, where $\tilde R = \C^{\N-1}$. The corresponding maximal isotropic $L = \im(Q)$ is given by $u \otimes S_-$.

The image of~$Q$ under the adjoint action of~$\fg_0$ consists of all elements of the form $u \otimes \tilde{U}$ and $\bar{u} \otimes v$. Acting on~$\fn_1$, $\ad_Q$ maps $S_- \otimes v^\vee $ isomorphically onto $L = u \otimes S_-$. Thus we have
\deq{
(\fn_Q)_1 = \left(\bar{u} \otimes \tilde{U}\right) \oplus \left(S_- \otimes \tilde{U}^\vee\right),
\quad
(\fn_Q)_2 = \bar{u} \otimes S_- \cong L^\vee,
}
with the bracket given by the evaluation pairing between $\tilde{U}$ and~$\tilde{U}^\vee$.
 The $R$-symmetry in the holomorphic twist is $\lie{gl}(\tilde{U})$, and the Lorentz symmetry is~$\lie{gl}(L^\vee)$.
 
 \numpar[chiral][Chiral superspace]
 In dimension zero modulo four, the fact that the pairing induced by Clifford multiplication is between the two chiral spinors of \emph{opposite} chirality---so that $\fn_1$ is a reducible representation of~$\fg_0$---allows one to perform a number of interesting constructions, notably including the \emph{chiral} versions of superspace familiar from the physics literature.
 We recall the constructions quickly here, in a manner that will generalize to the twist. Again, we follow the discussion in~\cite[Chapter 5, \S7]{Manin}.

\begin{dfn}
Let $\fn$ be the four-dimensional, $\N$-extended supersymmetry algebra whose structure we recalled in~\S\ref{4dnQ} above.
 The \emph{chiral subalgebras} $\fm^{(\pm)}$ are the abelian super Lie algebras
 \deq{
 \fm^{(+)} = S_+ \otimes R(-1) \oplus V(-2), \quad
 \fm^{(-)} = S_- \otimes R^\vee(-1) \oplus V(-2), \quad
 }
 mapping to~$\fn$ via the obvious inclusions.
 \end{dfn}
Working globally, a distribution $D$ specifying a superconformal structure of type $\fn$ on a supermanifold $X$
 is canonically a direct sum of two involutive subdistributions $D^{(\pm)}$ of type $\fm^{(\pm)}$. 
Each subdistribution integrates to a foliation with leaves of dimension $0|2\N$.

We can pass to the leaf space of~$D^{(-)}$, which we call $X^{(+)}$. Locally, this operation is modelled by replacing the sheaf of functions on~$X$ by the subsheaf of $D^{(-)}$-invariants.
If $p: X \to X^{(+)}$ denotes the corresponding projection map of supermanifolds, 
then $p^* T X^{(+)}$ is canonically identified with $TX/D^{(-)}$. Under this identification, the image of $D^{(+)}$ becomes an involutive distribution of maximal odd dimension. Thus $X^{(+)}$ is canonically a superspace of type $\fm^{(+)}$, called \emph{chiral superspace}, and the map $p$ is distribution-preserving. 

When $X$ is the flat superspace of type~$\fn$, the images under $p$ of the left-invariant vector fields (which give a distribution-preserving action of $\fn$ of~$N$) give rise to an action of~$\fn$ on~$X^{(+)}$ by smooth vector fields. Since $\fn^{(+)}$ is abelian, the condition of being distribution-preserving is vacuously satisfied on~$X^{(+)}$.

More generally, the map induced by $p$ from vector fields on~$X$ to vector fields on~$X^{(+)}$ is not a map of Lie algebras, since $\Gamma(D^{(-)})$ is not an ideal in~$\Vect(X)$.
But there \emph{is} a map from the Lie algebra of distribution-compatible vector fields on~$X$ to vector fields on~$X^{(+)}$. The former is, by definition, the normalizer in~$\Vect(X)$ of the subalgebra $\Gamma(D)$, and so is clearly contained in the normalizer of $\Gamma(D^{(-)}) \subseteq \Gamma(D)$.

\numpar[formulas][Formulas in coordinates]
For the convenience of the physics reader, we situate the previous discussion in the context of the typical notation for coordinates on four-dimensional $\N=1$ superspace. We use abstract indices $\alpha$ for a basis of~$S_+$ and $\dot\alpha$ for a basis of~$S_-$. Since $V \cong S_+ \otimes S_-$, a vector index is a pair of spinor indices, one dotted and one undotted.

The flat superspace $N = \exp(\fn)$ has even coordinates $y^{\alpha\dot\alpha}$ and odd coordinates $\theta^\alpha$ and~$\thetabar^{\dot\alpha}$. 
In an appropriate coordinate system, the right-invariant vector fields take the form 
\deq{
    D_\alpha = \pdv{ }{\theta^\alpha} - 2 \thetabar^{\dot\alpha} \pdv{ }{y^{\alpha\dot\alpha}}, \quad
    \bar{D}_{\dot\alpha} = \pdv{ }{\thetabar^{\dot\alpha}} ,
}
whereas the corresponding left-invariant vector fields are
\deq{
    Q_{\alpha} = \pdv{ }{\theta^{\alpha}} , \quad
    \bar{Q}_{\dot\alpha} = \pdv{ }{\thetabar^{\dot\alpha}} + 2 \theta^{\alpha} \pdv{ }{y^{\alpha\dot\alpha}}.
}
The kernel of~$D^{(-)}$ consists of functions that are of order zero in the generators $\bar\theta^{\dot\alpha}$. Taking the quotient, we find that the left-invariant vector fields are equivalent to 
\deq{
    p_*Q_{\alpha} = \pdv{ }{\theta^{\alpha}} , \quad
    p_*\bar{Q}_{\dot\alpha} =  2 \theta^{\alpha} \pdv{ }{y^{\alpha\dot\alpha}}
}
when acting on $\ker(D^{(-)})$. These are the normal formulas for the action of supersymmetry on chiral superspace. 
A \emph{chiral superfield} is an equivariant sheaf on~$N^{(+)}$; by pulling back along~$p$, we can view chiral superfields as particular examples of general superfields, which are equivariant sheaves on~$N$.

 
 \numpar[twistit][Chiral superspace for holomorphic twists]
In the holomorphic twist, $\fn_Q$ is still the sum of two involutive distributions, but these are no longer of the same dimension. We have
\deq{
\fm_Q^{(+)} = \tilde{U}(-1) \oplus L^\vee(-2), \quad
\fm_Q^{(-)} = (L^\vee \otimes \tilde{U}^\vee)(-1) \oplus L^\vee(-2).
}
By the discussion in~\S\ref{chiral} above, it is reasonable to expect a map from $\Conf(\fn)$ to $\Conf(\fm^{(+)})$, which in turn is identified with~$\Vect(\C^{2|\N-1})$. We think of $\Vect(\C^{2|\N-1})$ as being the symmetries of chiral superspace in the holomorphic setting. In each instance, we will see concretely that such a map exists.

\numpar[4dN=2hol][The cases $\N=2$ and $\N=3$]
The holomorphic twist of the $\cN=2$ algebra is equivalent to the component-field model $\mu \Conf(\fn_Q)$, which can be identified with the Dolbeault resolution of holomorphic sections of the holomorphic tangent bundle to $\C^{2|1}$. 
This is a dg Lie algebra with differential $\dbar$ and bracket which extends the usual bracket of holomorphic super vector fields. (The odd directions are treated as purely algebraic.)
We refer to~\cite{SCA} for a lengthy discussion of this example as an enhancement of the holomorphic twist of the usual supersymmetry algebra at the level of the holomorphic twist. Our results here prove that the higher Virasoro algebra arises directly as the holomorphic twist of the stress-tensor multiplet.

Similarly, the holomorphic twist of the $\cN=3$ algebra can be identified with the Dolbeault resolution of holomorphic vector fields on~$\C^{2|2}$. In each of these examples, the map from distribution-preserving vector fields on the full superspace to vector fields on the chiral superspace is in fact an isomorphism.

   \numpar[4dN=4hol][The case $\N=4$]

In this instance, $\tilde{U} = \C^3$, so that $(\fm_Q^{(+)})_1$ is three-dimensional and $(\fm_Q^{(-)})_1$ is six-dimensional. 
The component fields of $\mu \Conf(\fn_Q)$ are as follows:
\begin{itemize}[label={---}]
\item in homological degree zero, a copy of $\Vect(\C^{2|3})$, ;
\item in homological degree one, a copy of $\Pi \cO(\C^{2|3})$, starting in bidegree $(1,0)$.
\end{itemize}
Following the conventions of~\S\ref{sec: examples}, and letting $T$ denote the holomorphic tangent bundle of~$\C^2$, we can present the component fields in a table:
\begin{equation}
\begin{tikzcd}[column sep = 4 ex, row sep = 6 ex]
\Vectorstack{{T} { }} \ar[drrrr,style={out=290,in=110,looseness=0.3}]
& \Vectorstack{{\tilde{U}^\vee \otimes T} {\tilde{U}}} \ar[drr, style={out=310,in=130,looseness=0.3}]
& \Vectorstack{{\wedge^2 \tilde{U}^\vee \otimes T} {\tilde{U}^\vee \otimes \tilde{U}}} \ar[drr, style={out=330,in=150,looseness=0.3}]
& \Vectorstack{{\wedge^3 \tilde{U}^\vee \otimes T} {\wedge^2 \tilde{U}^\vee \otimes \tilde{U}}}
& \Vectorstack{{ } {\wedge^3 \tilde{U}^\vee \otimes \tilde{U}}} 
\\ 
& \ul{\C} & \tilde{U}^\vee & \wedge^2 \tilde{U}^\vee & \wedge^3 \tilde{U}^\vee
\end{tikzcd}
\end{equation}
We emphasize that this is \emph{not} the algebra of divergence-free vector fields on~$\C^{2|3}$. The only possible differentials that are equivariant for the $R$-symmetry are indicated in the diagram. They are second- and third-order differential operators, consisting either of the Laplacian or of the Laplacian following the divergence operator.

\subsection{Holomorphic twists in six dimensions}
\label{ssec:tw6d}

Motivated by the traditional classification of superconformal algebras, we address only the case of $(\N,0)$ chiral supersymmetry. Twisting supercharges always exist, and are unique up to equivalence for $\N=(1,0)$. In this case, the twisting supercharge is holomorphic, and determines a complex structure on~$V = \R^6$ via the maximal isotropic $ L = \im(\ad_Q) \cong \C^3 \subset V_\C$. 

\numpar[p:basix][The supertranslation algebra $\fn_Q$] 
The odd elements of the supertranslation algebra take the form
\deq{
\fn_1 = S_+ \otimes U,
}
where $U = \C^{2\N}$ is equipped with a nondegenerate symplectic pairing $\omega$. The even elements are $\fn_2 = V \cong \wedge^2 S_+$.
The Lie algebra of automorphisms is $\fg_0 = \lie{so}(V) \oplus \lie{gl}(1) \oplus \lie{sp}(U)$; we recall the exceptional isomorphism $\lie{so}(\wedge^2 S_+) \cong \lie{sl}(S_+)$.

Without loss of generality, $Q$ is again of the form $u \otimes v$ for a pair of highest-weight vectors. Then $S_+ = \C\cdot u \oplus L$ and $U = \C\cdot v \oplus \C\cdot \bar{v} \oplus \tilde{U}$, where $\tilde{U} = \C^{2\N-2}$ is again symplectic. (We can identify it with the symplectic reduction $U \sympred \C\cdot v$.) By abuse of notation, we identify $L \subset S_+$ with its  image $u \wedge S_+ \subset \wedge^2 S_+ \cong V$.

The image of~$\fg_0$ under~$\ad_Q$ consists of all elements of the form $u \otimes \tilde{U}$, $u \otimes \bar{v}$, and $L \otimes v$. Acting on~$\fn_1$, $\ad_Q$ maps $L \otimes \bar{v}$ isomorphically onto $L \subset V$. Thus we have that
\deq{
(\fn_Q)_1 = L \otimes \tilde{U}, 
\quad
(\fn_Q)_2 = L^\vee \cong \wedge^2 L.
}
The bracket is given by the wedge product on~$L$ and the symplectic contraction on~$\tilde{U}$. The $R$-symmetry in the holomorphic twist is $\lie{sp}(\tilde{U})$, and the Lorentz symmetry is~$\lie{sl}(L^\vee)$. 

\numpar[eg:6hol][The case $\N=(2,0)$]
In this case, the canonical multiplet is the abelian tensor multiplet. The twisted super Poincar\'e algebra and the twist of the tensor multiplet were computed in~\cite{spinortwist,twist20}.
The holomorphic twist of the canonical multiplet $A^\bu(\cO_{Y})^{Q_{\hol}} \simeq A^\bu(\cO_{Y_{Q_{\hol}}})$ admits the following simple description as a complex of sheaves on $\C^3$:
\begin{equation}
	\begin{tikzcd}[row sep=0.3cm, column sep=0.2cm]
	\Omega^{0,\bu}  \ar[dr,"\partial"] & \\
	& \Omega^{1,\bu} & \Omega^{0,\bu}(\Pi \tilde{U}) \otimes K^{1/2}.
	\end{tikzcd}
\end{equation}
As discussed at length in~\cite{twist20}, the theory describes a system related to the intermediate Jacobian together with holomorphic symplectic bosons.

Since $\tilde{U} = \C^2$ here, 
the description above makes it clear that the nilpotence variety of~$\fn_Q$ can be identified with the variety of two-by-three matrices with rank less or equal to one. The defining ideal $I$ is spanned by the three two-by-two minors. The projective version of the nilpotence variety is thus the Segre embedding $\P^1 \times \P^2 \subset \P^5$. We can thus understand all line bundles and the corresponding multiplets by techniques analogous to those used in~\cite{6dbundles}.
We denote the projection maps on the first and second factors by $\pi_{1,2}$ respectively.
All equivariant line bundles on $\P^1 \times \P^2$ can be obtained via pullback from the respective projective factors. Each is isomorphic to exactly one of the lines
\begin{equation}
	\cO(m,n) := \pi_1^* \cO_{\P^1}(m) \otimes_{\cO_{\P^1 \times \P^2}} \pi^*_2 \cO_{\P^2}(n).
\end{equation}
There is a corresponding family of multiplets, obtained by applying the pure spinor construction to the 
graded global section module 
\begin{equation}
	\Gamma_{(m,n)} = \Gamma_*(\cO(m,n)) =  \bigoplus_{k \in \ZZ} H^0(\cO(m+k,n+k)) .
\end{equation}

We now describe the holomorphic twist of $\Conf(\lie{n}_Q)$.
Like the canonical multiplet, this multiplet is associated to a line bundle on the nilpotence variety.
Indeed, Theorem~\ref{thm: coker} implies that $\Conf(\lie{n}_Q)$ arises from the cokernel of the map
\begin{equation}
\varphi: (L \otimes \tilde{U}) \otimes R/I \to L^\vee \otimes R/I
\end{equation}
induced by the bracket. This module is isomorphic to $\Gamma_{(0,1)}$: one can see easily that the weight-zero piece transforms in~$L^\vee$, which is isomorphic to~$H^0(\cO(0,1))$.

For the component fields of the multiplet, one recovers the local Lie algebra
\begin{equation}~\label{eq: comp-der}
	\mu \Conf(\lie{n}_{\hol}) = 
	\begin{bmatrix}
	\Omega^{0,\bu}(T) & \tilde{U} \otimes \Omega^{1,\bu} & \lie{sp}(\tilde{U}) \otimes \Omega^{0,\bu} 
	\end{bmatrix}.
\end{equation}
This is supported in degree zero at the holomorphic level, and thus acquires a strict Lie structure. On flat space, we precisely recover the exceptional infinite-dimensional super Lie algebra $E(3|6)$~\cite[\S4.4]{ChengKac}. Our results thus prove that $E(3|6)$ is the holomorphic twist of the $\N=(2,0)$ stress tensor multiplet.

\printbibliography

\end{document}